\PassOptionsToPackage{hyphens}{url}
\RequirePackage{fix-cm}
\documentclass[twocolumn,natbib]{svjour3}          
\smartqed  
\usepackage{graphicx}
 \usepackage{mathptmx}      
\usepackage{algorithm,amsfonts,amsmath,enumitem,orcidlink}
\usepackage[caption=false]{subfig}  

\usepackage{pgf, tikz}
\usetikzlibrary{arrows, automata}

\DeclareMathOperator{\var}{var}

\newcommand{\cond}{\mid}

\spnewtheorem{assumption}{Assumption}{\bf}{\rm}

\newcommand{\add}[1]{{#1}}  
\newcommand{\ad}[1]{{#1}}  

\journalname{Statistics and Computing}
\begin{document}

\title{Unbiased approximation of posteriors via coupled particle Markov chain Monte Carlo
}


\author{
Willem van den Boom\orcidlink{0000-0002-1777-3857}
\and
Ajay Jasra\orcidlink{0000-0003-4808-9131}
\and
Maria De Iorio\orcidlink{0000-0003-3109-0478}
\and
Alexandros Beskos\orcidlink{0000-0002-4577-8153}
\and
Johan G.\ Eriksson\orcidlink{0000-0002-2516-2060}
}


\institute{
\ad{Willem van den Boom \and Maria De Iorio \and Johan G.\ Eriksson} \at
\ad{National University of Singapore, Yong Loo Lin School of Medicine} \\
              \email{\ad{vandenboom@nus.edu.sg}}           
           \and
           Willem van den Boom \and Maria De Iorio \and Johan G.\ Eriksson \at
           Agency for Science, Technology and Research, Singapore Institute for Clinical Sciences
           \and
           Ajay Jasra \at
           King Abdullah University of Science and Technology, Computer, Electrical and Mathematical Sciences and Engineering division, Thuwal, Saudi Arabia
           \and
           Maria De Iorio \and Alexandros Beskos \at
           University College London, Department of Statistical Science, UK
}

\date{Received: date / Accepted: date}

\maketitle

\begin{abstract}
Markov chain Monte Carlo (MCMC) is a powerful methodology for the approximation of posterior distributions. However, the iterative nature of MCMC does not naturally facilitate its use with modern highly parallel computation on HPC and cloud environments. Another concern is the identification of the bias and Monte Carlo error of produced averages. The above have prompted the recent development of fully (`embarrassingly') parallel unbiased Monte Carlo methodology based on coupling of MCMC algorithms. A caveat is that formulation of effective coupling is typically not trivial and requires model-specific technical effort. We propose coupling of \add{MCMC chains deriving from} sequential Monte Carlo (SMC) by considering adaptive SMC \add{methods in combination} with recent advances in unbiased estimation for state-space models. Coupling is then achieved at the SMC level and is, \add{in principle}, not problem-specific. The resulting methodology enjoys desirable theoretical properties. \add{A central motivation is to extend unbiased MCMC to more challenging targets compared to the ones typically considered in the relevant literature.} We illustrate the effectiveness of the algorithm via application to two \add{complex} statistical models: (i) horseshoe regression; (ii) Gaussian graphical models.
\end{abstract}

\subsection*{Declarations}

\subsubsection{Funding}

This work is supported by the Singapore Ministry of Education Academic Research Fund Tier~2 (grant number MOE2019-T2-2-100)
and the Singapore National Research Foundation under its Translational and Clinical Research Flagship Programme and administered by the Singapore Ministry of Health’s National Medical Research Council (grant number NMRC/TCR/004-NUS/2008; NMRC/TCR/012-NUHS/2014). Additional funding is provided by the Singapore Institute for Clinical Sciences, Agency for Science, Technology and Research.

\subsubsection{Conflicts of interest/Competing interests}

The authors have no conflicts of interest to declare that relate  to the content of this article.

\subsubsection{Availability of data and material}

The data are confidential human subject data, thus are not available.

\subsubsection{Code availability}

The scripts that produced the empirical results are available on \url{https://github.com/willemvandenboom/cpmcmc}.

\section{Introduction}
\label{sec:intro}

MCMC is a powerful methodology for the  approximation of  complex distributions.
MCMC is intrinsically iterative and, while asymptotically unbiased, the size of the bias and the  Monte Carlo error of generated estimates given a finite number of \add{iterations} are \add{often} difficult to quantify.
\add{Moreover, MCMC will typically} not allow full exploitation of the computational potential 
of modern distributed-computing techniques. Recently, 
\citet{Jacob2020_2} propose a method for unbiased MCMC estimation based on coupling of Markov chains, building on ideas by \citet{Glynn2014}.
The algorithm is  embarrassingly parallel, and the unbiasedness  provides immediate quantification of the Monte Carlo error.
However, devising effective coupling for MCMC algorithms targeting a given posterior can be highly challenging. See e.g.~the construction of coupled MCMC for a horseshoe regression model in \citet{Biswas2020}
and the development for posteriors based on Hamiltonian Monte Carlo (HMC) in \citet{Heng2019}.

\add{Considering the scope for unbiased MCMC and its current limitations,
we propose a coupled MCMC algorithm where the coupling mechanism is not in principle specific to
the posterior at hand,
resulting in a general recipe for unbiased MCMC.
We do so
by working on an appropriate augmented space.
Specifically,
we build on recent advances for unbiased estimation in state-space models
\citep{Middleton2019,Jacob2020}
which devise coupling strategies
through particle filters independently of the shape of the target distribution.
Our work extends these ideas   by embedding a general posterior distribution in a state-space model using adaptive SMC \citep{DelMoral2006}.
This results in a methodology that broadens the class of posteriors  that can be treated via unbiased MCMC, by exploiting the coupling methods feasible within the SMC framework.
As our methodology couples Markov chains that consist of particle MCMC methods, we refer for convenience to our  method as `coupled particle MCMC'.
To illustrate the potential of our methodology, we apply coupled particle MCMC to (i) Gaussian graphical models, which are substantially more complex than models considered previously in the literature on unbiased MCMC, mainly due to discreteness and dimension of graph space; (ii) horseshoe regression.

The MCMC step resulting from our state space embedding combines particle MCMC methods by \citet{Middleton2019} and \citet{Jacob2020}, namely coupled particle independent Metropolis-Hastings (PIMH) and conditional SMC, respectively.
The focus of these works is unbiased approximation for the smoothing distribution of a state-space model; the high-dimensionality therein relates to the number of time steps, whereas the state space is implicitly assumed to be low-dimensional.
We consider more complex state spaces of a much higher dimension than \citet{Middleton2019} and \citet{Jacob2020}, provide effective adaptation of the SMC,
and 
\ad{consider both}
PIMH and conditional SMC
to improve mixing
(as compared to using only PIMH)
of the MCMC. 
We empirically investigate the trade-off between mixing and coupling via a number of numerical experiments.}

The structure of the paper is as follows.
\add{Section~\ref{sec:coupled} reviews unbiased MCMC estimation based on coupled Markov chains.}
Section~\ref{sec:cpmcmc} introduces coupled particle MCMC for unbiased estimation, for the case of a \add{general} posterior.
Section~\ref{sec:theory} contains theoretical results for the meeting time \add{of the coupled Markov chains} and the \add{resulting}
unbiased estimator
deduced from the literature on coupled conditional SMC for state-space models \citep{Lee2020}.
Section~\ref{sec:simul} applies the proposed general methodology to simulated data, including a case with horseshoe regression, where a comparison with the MCMC coupling  approach in \citet{Biswas2020} is carried out.
Section~\ref{sec:ggm} considers Gaussian graphical models as an example where effective coupling of MCMC is \add{highly non-trivial. We finish with a discussion in Section~\ref{sec:discussion}.

\section{Unbiased MCMC with couplings}
\label{sec:coupled}

Before we introduce coupled particle MCMC,
we review coupled MCMC and introduce notation.}
Denote the parameter space by $\mathcal{X}\subseteq \mathbb{R}^{d_x}$, $d_x\ge 1$.
Let $x\in\mathcal{X}$ denote the parameter, 
$y\in\mathcal{Y}\subseteq \mathbb{R}^{d_y}$, $d_y\ge 1$, the data, 
and $\pi(x)$ the density of the posterior of interest w.r.t.~some dominating measure.
The unbiased construction requires a pair of coupled ergodic Markov chains 
$\{x(t-1),\bar{x}(t)\}$, $t \geq 1$, on $\mathcal{X}\times \mathcal{X}$ with both chains having $\pi(x)$ as equilibrium distribution.
The coupling is such that $x(t-1)$ and $\bar{x}(t)$ have the same distribution for any $t\geq 1$, and
the two chains meet at some random  time $\tau<\infty$ a.s.\ so that $x(t)=\bar{x}(t)$ for $t\geq\tau$.
Under standard conditions, the posterior expectation $\pi(h) = \int_{\mathcal{X}}h(x)\,\pi(x)\,dx$ of a statistic $h:\mathcal{X}\mapsto \mathbb{R}$ can be obtained as
$\pi(h) = \lim_{t\to\infty} E\left[ h\{x(t)\}\right]$, \add{with} all expectations assumed finite.
Writing the limit as a telescoping sum and using that fact that $x(t-1)$, $\bar{x}(t)$ admit the same law, 
gives for any fixed $k\geq 0$  \citep{Glynn2014}
\begin{align*}
\pi(h) &= E[h\{x(k)\}] + \sum_{t=k+1}^{\infty} \left(E[h\{x(t)\}] - E[h\{x(t-1)\}]\right) \\
	&= E[h\{x(k)\}] + \sum_{t=k+1}^{\infty} \left(E[h\{x(t)\}] - E[h\{\bar{x}(t)\}]\right) \\
&= E\left(h\{x(k)\} + \sum_{t=k+1}^{\tau-1} \left[h\{x(t)\} - h\{\bar{x}(t)\}\right] \right).
\end{align*}
We assume that the technical conditions that permit the exchange of summation and expectation in the last step hold in our setting.
Thus, the quantity 
\begin{equation} \label{eq:unbiased_simple}
\hat{h}_k
= h\{x(k)\}
	+ \sum_{t=k+1}^{\tau - 1} [h\{x(t)\} - h\{\bar{x}(t)\}]
\end{equation}
is an unbiased estimator of $\pi(h)$. Process $\{\bar{x}(t)\}$ is typically  initialised at $\bar{x}(1)=x(0)\sim p_0$, for \add{a} law $p_0$.
Both chains will evolve marginally according to
the same MCMC chain with the posterior $\pi(x)$ as invariant distribution, with a coupling applied for the joint transition 
$[\,x(t),\bar{x}(t)\cond x(t-1),\bar{x}(t-1)\,]$, $t\ge 2$. We adopt this setting for the rest of the paper.

\add{Coupled particle MCMC requires some further notation.}
Denote
the prior density on the parameter by $p(x)$.
Let $p(y\cond x)$ be the density of the data $y$ given~$x$.
The posterior density follows from Bayes' rule as $\pi(x) \propto p(x)\,p(y\cond x)$.
For any  $\alpha\in[0, 1]$, we denote by $\pi_\alpha(x)$ the tempered posterior proportional to $p(x)\,p(y\cond x)^\alpha$.
Thus, $\pi_0(x) = p(x)$ and $\pi_1(x) = \pi(x)$.
All densities are assumed to be determined w.r.t.~appropriate reference measures.
As our method involves adaptive SMC, we assume we can sample $x$ from its prior $p(x)$ and evaluate the likelihood $p(y\cond x)$.
The method requires, for any $\alpha\in(0, 1]$, the construction of an MCMC step with $\pi_\alpha(x)$ as its invariant distribution.
\add{We refer to this step as the `inner' MCMC step as it is a part of a more involved `outer' particle MCMC step.}
For integers $i\leq j$, denote the range $\{i,\dots,j\}$ by ${i\! :\! j}$.
We use the colon notation for collections of variables,
i.e.,
$x_{i:j} = \{x_i,\dots,x_j\}$
and
$x^{i:j} = \{x^i,\dots,x^j\}$.

\section{Coupled particle MCMC}
\label{sec:cpmcmc}

\subsection{Feynman-Kac model \& SMC sampler}
\label{sec:smc}

The proposed unbiased estimation procedure builds on SMC.
As a first step, 
we `adapt' the SMC algorithm to the prior $p(x)$ and likelihood $p(y\cond x)$ under consideration.
This is a preliminary step that determines the tempering constants and the inner MCMC steps, 
thus \add{also} the target Feynman-Kac model.
\add{By exploiting adaptive SMC, we are able to obtain a flexible and general \ad{construction} (i.e. applicable to a large class of posterior distributions) for unbiased posterior approximation.}

The adaptation produces a sequence of 
 $S\ge 0$ temperatures, $0<\alpha_1<\cdots < \alpha_S <1$, 
corresponding to 
bridging distributions
$\pi_{\alpha_s}(x)$, $s=1,\dots,S$.
Here,
$S$ and
temperatures
$\alpha_s$, $s=1,\dots,S$, are chosen from this initial application of SMC with $N_0\ge 1$ particles, 
so that importance weights meet an effective sample size (ESS) threshold as, e.g., in \citet{Jasra2010}.
The adaptation also \add{determines} the number $m_s$, $s=1,\ldots, S$, of iterations of \add{an `inner'} MCMC kernel $p_{\alpha_s}(dx'\mid x)$ that preserves $\pi_{\alpha_s}(x)$. \add{More in details,}
for each $s$, we choose $m_s$ via a criterion that requires sufficiently reduced sample correlation, for \add{particles} pre- and post-\add{application of MCMC steps, over given} scalar statistics \add{of interest,} $f_j:\mathcal{X}\to \mathbb{R}$, $j=1,\ldots $.
See Sections~\ref{sec:simul} and \ref{sec:ggm_results} for examples of such statistics for specific models, which in both cases include the log-likelihood $f_1(x) = \log\{p(y\cond x)\}$.
\add{Concatenating} $m_s$ \add{inner} MCMC steps increases diversity among the particles
$x^{1:N_0}_s$ to avoid weight degeneracy.
\citet{Kantas2014}
consider similar adaptation of $m_s$.
Algorithm~\ref{alg:FK} summarizes the adaptive procedure
where $\widehat{\mathrm{corr}}(\cdot, \cdot)$ in Step~\ref{step:corr} denotes the sample correlation.

\begin{algorithm}
\caption{Adaptation of the Feynman-Kac model. \label{alg:FK}}
Input: Number of particles $N_0$, ESS and correlation thresholds $\gamma_0, \zeta_0 \in [0,1]$.
\begin{enumerate}
	\item \label{step:init}
	Sample particles $x_0^{1:N_0}$ independently from $p(x)$. Set $s=1$, $\alpha_0=0$.
	\item \label{step:repeat}
	Repeat while $\alpha_{s-1}<1$:
	\begin{enumerate}
		\item Compute weights $w_s^i(\alpha) = p(y\cond x^i_{s - 1})^{\alpha - \alpha_{s - 1}}$, $i=1,\dots,N_\add{0}$, 
and find 
\begin{align*}
\alpha_s = \min \left\{ \alpha\in(\alpha_{s-1},1] : 1/\sum_{i=1}^{N_0} \{w_s^i(\alpha)\}^2 \le \gamma_0 N_0 \right\}.
\end{align*}
\item \add{Obtain} $x^{1:N_0}_s$ by sampling with replacement from $x^{1:N_0}_{s-1}$
		with probabilities proportional to $w_s^{1:N_0}(\alpha_s)$. 
\item \label{step:corr} Let $x^{1:N_0}_{s,k}$ denote the position of particles $x^{1:N_0}_s$ after applying $k\ge 1$ MCMC transitions 
$p_{\alpha_s}(dx'\mid x)$, on each particle.
		Find
\[
m_s = \min_{k\ge 1}\Big\{ \max_j \big[ \widehat{\mathrm{corr}}\{f_j(x_{s}^{1:N_0}),f_j(x_{s,k}^{1:N_0})\}  \big] \le \zeta_0  \Big\}.
\]
		With an abuse of notation, let $x^{1:N_0}_s$ now be the particles after the application of $m_s$ MCMC steps. 
\item Set $s=s+1$.
	\end{enumerate}
\end{enumerate}
Output: Temperatures $0<\alpha_1<\cdots < \alpha_S <1$, number of  MCMC steps $m_{1:S}$.
\end{algorithm}

\begin{algorithm}
\caption{SMC sampler. \label{alg:smc}}
Input: Number of particles $N$, temperatures $0= \alpha_0<\dots<\alpha_S < 1$, number of MCMC steps $m_{1:S}$.
\begin{enumerate}
	\item \label{step:smc_init}
	Sample $N$ particles $x^{1:N}_0$ independently from $\pi_{\alpha_0}(x) = p(x)$.
	\item
	For $s=1,\dots,S$:
	\begin{enumerate}
		\item \label{step:weight1}
		Compute weights $w_s^i = p(y\cond x^i_{s - 1})^{\alpha_s - \alpha_{s - 1}}$, $i=1,\dots,N$.
		\item \label{step:resample}
		Determine $x^{1:N}_s$ by sampling with replacement from $x^{1:N}_{s-1}$
		with probabilities proportional to $w_s^{1:N}$. 
		\item \label{step:mcmc}
		For each particle in $x^{1:N}_s$,
		\add{carry out} $m_s$ MCMC transitions $p_{\alpha_s}(dx'\mid x)$.
With an abuse of notation, let $x^{1:N}_s$ be the particles after application of the MCMC steps. 
	\end{enumerate}
	\item
	\begin{enumerate}
		\item \label{step:weight2}
		Compute weights $w^i = p(y\cond x^i_S)^{1 - \alpha_S}$, $i=1,\dots,N$.
		\item \label{step:resample2}
		Determine $x^{1:N}$ by sampling with replacement from $x^{1:N}_S$
		with probabilities proportional to $w^{1:N}$.
	\end{enumerate}
\end{enumerate}
Output: Set of particles $x^{1:N}$ that approximate the posterior $\pi(x)$. 
\end{algorithm}

Algorithm~\ref{alg:FK} determines the Feynman-Kac model, \add{i.e.~a distribution $\Pi(x_{0:S})$ with marginal $x_S\sim \pi(x)$.} Then, Algorithm~\ref{alg:smc} describes a standard 
SMC sampler \add{applied to such model} corresponding to a bootstrap particle filter.
Steps~\ref{step:resample} and \ref{step:resample2} of Algorithm~\ref{alg:smc}
describe multinomial resampling.
For our empirical results, we replace multinomial with systematic resampling, as the latter 
reduces variability
\citep[Section~9.7]{Chopin2020}
and yields 
better mixing for the outer MCMC steps.
Algorithm~\ref{alg:sys_sampling}, 
\add{in} Appendix~\ref{ap:sys_sampling},
describes systematic resampling.
\add{Use of} adaptive resampling \citep[Section~10.2]{Chopin2020}
further reduces the \add{variability} of Monte Carlo estimates in our empirical results.
That is,
we only resample \add{(Step~\ref{step:resample})}
if the ESS of the current weighted particle approximation
falls below $N\gamma$ for a $\gamma\in [0,1]$.

Step~\ref{step:resample2} of Algorithm~\ref{alg:smc} is not required to approximate the posterior ${\pi(x)}$
as the pair $(w^{1:N}, x_S^{1:N})$ provides a weighted approximation.
We include the step since conditional SMC will involve it.
Section~\ref{sec:unbiased} discusses \add{a Rao-Blackwellization approach enabling} the use of the weighted approximation \add{within} coupled particle MCMC.

\add{Algorithms~\ref{alg:FK} and \ref{alg:smc}
use different numbers of particles, namely $N_0$ and $N$, respectively.
We choose $N_0$ to be larger than $N$ as the the number $N_0$ is used once, `offline', and results provided by this single run of Algorithm~\ref{alg:FK} will fix all aspects of the Feynman-Kac model used by coupled particle MCMC. Choice of a large enough $N_0$ aims to facilitate stability for the obtained collection of temperatures and number of inner MCMC steps.}

\subsection{Coupling of particle MCMC}
\label{sec:ccsmc_sub}

Having specified an SMC sampler,
we derive a coupled particle MCMC algorithm.
The outer MCMC step is constructed \add{via} the SMC sampler.
Specifically, we borrow ideas  from the particle filtering literature and define 
a coupling strategy of the \add{outer} MCMC step, \add{the latter defined} as  a mixture of
the coupled PIMH in \citet[Algorithm~3]{Middleton2019}
and
the coupled conditional particle filters in \citet[Algorithm~2]{Jacob2020}.
PIMH and conditional SMC \citep[Section~2.4]{Andrieu2010} provide MCMC steps on the extended state space $x_{0:S}\in\mathcal{X}^{S+1}$ based on Algorithm~\ref{alg:smc}.
The invariant law on $\mathcal{X}^{S+1}$ of such MCMC has ${\pi(x)}$ as marginal distribution on $x_S\in\mathcal{X}$ \citep[Theorems~2, 5]{Andrieu2010}.
Therefore, the resulting MCMC algorithm can be run for two coupled chains to provide unbiased Monte Carlo estimation per \eqref{eq:unbiased_simple}.
The added machinery of SMC provides more ways to couple the MCMC compared to using a less elaborate MCMC algorithm for sampling from the posterior $\pi(x)$.

Coupled PIMH results in smaller meeting times $\tau$
and worse mixing of the MCMC chain
than conditional SMC in our experiments.
The meeting times and the MCMC mixing both affect the variance of the resulting unbiased estimators.
The empirical results show that neither PIMH nor conditional SMC yields universally better performance \ad{(across different posteriors).
We thus also consider a mixture of them as the outer MCMC step. 
We mention here that in our experiments we do not encounter scenarios where such mixture yields notably more computationally efficient unbiased estimation than both PIMH and conditional SMC.
Still, incorporation of the mixture allows for gaining further insight into the presented methods.}
Next, we describe coupled PIMH and coupled conditional SMC separately.

\begin{algorithm}
\caption{Coupled PIMH step. \label{alg:cpimh}}
Input: Current states $x_{0:S}$ and $\bar{x}_{0:S}$, \add{along with the} corresponding SMC estimates $Z$ and $\bar{Z}$ of the marginal likelihood $p(y)$.
\begin{enumerate}
	\item \label{step:pimh_proposal}
	Sample $\tilde{x}_{0:S}$ and $\tilde{Z}$ using Algorithm~\ref{alg:smc} as the proposal for both chains:
	\begin{enumerate}
		\item
		Set $\tilde{x}_S$ equal to $x^1$ from the output of Algorithm~\ref{alg:smc}.
		\item
		For $s = S-1,\dots,0$, set $\tilde{x}_s$ equal to the element in $x^{1:N}_s$ which generated $\tilde{x}_{s+1}$ per Step~\ref{step:resample} of Algorithm~\ref{alg:smc}.
		In other words, $\tilde{x}_s$ is the ancestor of $\tilde{x}_{s+1}$.
    	\item
    	Compute the corresponding marginal likelihood estimate $\tilde{Z}$ from the weights in Algorithm~\ref{alg:smc} as
    	detailed in \citet[Section~3.2.1]{DelMoral2006}.
	\end{enumerate}
	\item \label{step:pimh_accept}
	Perform two coupled Metropolis-Hastings accept-reject steps:
	\begin{enumerate}
    	\item
    	Sample $U\sim\mathcal{U}(0,1)$.
    	\item
    	If $U < \tilde{Z}/Z$, then set $x_{0:S} = \tilde{x}_{0:S}$ and $Z = \tilde{Z}$.
    	\item
    	If $U <  \tilde{Z}/\bar{Z}$, then set $\bar{x}_{0:S} = \tilde{x}_{0:S}$ and $\bar{Z} = \tilde{Z}$.
	\end{enumerate}
\end{enumerate}
Output: Updated states $x_{0:S}$ and $\bar{x}_{0:S}$, \add{along with $Z$ and $\bar{Z}$,
obtained by application of a coupled PIMH kernel, with transitions that marginally preserve $\Pi(x_{0:S})$.}
\end{algorithm}

\subsubsection{Coupled PIMH}

Algorithm~\ref{alg:cpimh} details the coupled PIMH update \add{which builds on the SMC sampler in Algorithm~\ref{alg:smc}.}
SMC provides an
unbiased estimate $Z$ of the marginal likelihood \add{$p(y)$
which 
enables the use of an independent Metropolis-Hastings step for two chains. 
The two steps are coupled} by using the same proposal and the same uniform random variable $U$ in the accept-reject step
across both chains.
Then, both chains meet as soon as they both accept the proposal in Step~\ref{step:pimh_accept} at the same MCMC iteration.
\add{See \citet{Middleton2019} for a more elaborate introduction of coupled PIMH.}

\begin{algorithm}
\caption{Coupled conditional SMC step. \label{alg:ccsmc}}
Input: Current states $x_{0:S}$ and $\bar{x}_{0:S}$.
\begin{enumerate}
	\item
	Set $x^1_s = x_s$ and $\bar{x}^1_s = \bar{x}_s$
	for $s=0,\dots,S$.
	\item \label{step:ccsmc_init}
	Sample $N-1$ particles $x^{2:N}_0$ independently from $\pi_{\alpha_0}(x)=p(x)$.
	\item
	Set $\bar{x}^i_0 = x^i_0$ for $i=2,\dots,N$.
	\item
	For $s=1,\dots,S$:
	\begin{enumerate}
		\item
		Compute \add{the (unnormalised)} weights $w_s^i = p(y\cond x_{s-1}^i)^{\alpha_s - \alpha_{s - 1}}$, \add{and}  $\bar{w}_s^i = p(y\cond \bar{x}_{s-1}^i)^{\alpha_s - \alpha_{s - 1}}$,
		$i=1,\dots,N$.
		\item \label{step:coupled_resampling}
		Determine $x^{2:N}_s$ and $\bar{x}^{2:N}_s$ by coupled resampling (Algorithm~\ref{alg:coupled_sampling}) with replacement from $x^{1:N}_{s-1}$
		and $\bar{x}^{1:N}_{s-1}$
		with probabilities proportional to $w_s^{1:N}$ and $\bar{w}_s^{1:N}$, respectively.
		\item
		\label{step:cmcmc}
		Update $x^i_s$ and $\bar{x}^i_s$ by \add{applying} $m_s$ coupled \add{inner} MCMC steps \add{that preserve (marginally)} $\pi_{\alpha_s}(x)$ for $i=2,\dots,N$.
		\item \add{Form trajectory $x_{0:s}^{i}$ by joining the ancestors of $x_s^{i}$, $i=1,\ldots, N$.}
	\end{enumerate}
	\item
	\begin{enumerate}
		\item
		Compute \add{the} weights $w^i = p(y\cond x_S^i)^{1 - \alpha_S}$, $\bar{w}^i = p(y\cond \bar{x}_S^i)^{1 - \alpha_S}$,
		$i=1,\dots,N$.
		\item \label{step:final_resample}
		Determine $x_{0:S}$ and $\bar{x}_{0:S}$ by coupled resampling (Algorithm~\ref{alg:coupled_sampling}) with replacement from $\{x^i_{0:S}\}_{i=1}^N$ and $\{\bar{x}^i_{0:S}\}_{i=1}^N$
		with probabilities proportional to
		$w^{1:N}$ and $\bar{w}^{1:N}$, respectively.
		\ad{\item
    	Compute the corresponding marginal likelihood estimates $Z$ and $\bar{Z}$ from the weights $w_{1:S}^{1:N}$ and $\bar{w}_{1:S}^{1:N}$, respectively, as
    	detailed in \citet[Section~3.2.1]{DelMoral2006}.}
	\end{enumerate}
\end{enumerate}
Output: Updated states $x_{0:S}$ and $\bar{x}_{0:S}$,
\ad{along with $Z$ and $\bar{Z}$,} \add{obtained
by application of a coupled conditional SMC kernel, with transitions that marginally 
preserve $\Pi(x_{0:S})$.}
\end{algorithm}

\begin{algorithm}
\caption{Coupled resampling. \label{alg:coupled_sampling}}
Input: Probability vectors $p_{1:N}$ and $\bar{p}_{1:N}$.
\begin{enumerate}
	\item
	Compute $p^{\min}_i = \min(p_i, \bar{p}_i)$ for $i = 1,\dots,N$.
	\item
	\begin{enumerate}
    	\item
    	With probability $a = \sum_{i=1}^N p^{\min}_i$:
    	\begin{enumerate}
    		\item
    		Draw $i$ according to the law on $1\! : \! N$ defined by the probability vector $p^{\min}_{1:N}/a$.
    		\item
    		Let $\bar{i} = i$.
    	\end{enumerate}
    	\item
    	With probability $1-a$,
    	draw $i$ and $\bar{i}$ according to the laws defined by the probability vectors
    	$(p_{1:N} - p^{\min}_{1:N})/(1-a)$ and ${(\bar{p}_{1:N} - p^{\min}_{1:N})}/{(1-a)}$, respectively.
	\end{enumerate}
\end{enumerate}
Output:
Samples $i$ and $\bar{i}$ which are distributed according to $p_{1:N}$ and $\bar{p}_{1:N}$, respectively,
and for which the probability of $i=\bar{i}$ is maximized.
\end{algorithm}

\subsubsection{Coupled conditional SMC}

Algorithm~\ref{alg:ccsmc} details the coupled conditional SMC update \add{which, like Algorithm~\ref{alg:cpimh}, builds on the SMC sampler in Algorithm~\ref{alg:smc}.}
The resampling of the particles in Steps~\ref{step:coupled_resampling} and \ref{step:final_resample}
forms the main source of coupling across the two \add{transition kernels applied on $x_{0:S}$ and $\bar{x}_{0:S}$.
See \citet{Jacob2020} for a more elaborate introduction of an algorithm closely related to Algorithm~\ref{alg:ccsmc}, namely the coupled conditional particle filter.}

As in \citet{Jacob2020} and \citet{Lee2020},
we consider Algorithm~\ref{alg:coupled_sampling}
for the coupled resampling.
Earlier applications of Algorithm~\ref{alg:coupled_sampling} can be found in
\citet{Chopin2015} \add{in the context of} theoretical analysis of conditional SMC
and \add{in} \citet{Jasra2017}
\add{within the setting of} multilevel particle filtering.
The algorithm samples from two discrete distributions such that the resulting two indices are equal with maximum probability \citep[Section~3.2]{Jasra2017}.

We use systematic resampling (\add{see} Algorithm~\ref{alg:sys_sampling}, \add{in} Appendix~\ref{ap:sys_sampling}) for \add{our} empirical results. Thus, we require a coupling for it.
\citet{Chopin2015} derive a modification of systematic resampling
for use with conditional SMC.
We propose a coupling of this conditional systematic resampling
for Step~\ref{step:coupled_resampling} of
Algorithm~\ref{alg:ccsmc}.
Appendix~\ref{ap:sys_sampling} \add{gives the details} for these algorithms.

Step~\ref{step:cmcmc} of Algorithm~\ref{alg:ccsmc} involves a coupling of the \add{inner} MCMC update for ${\pi_{\alpha_s}(x)}$ across two chains.
The coupling should at least be `faithful' \citep{Rosenthal1997}, i.e., sustain any meeting of the chains.
That is, if $x^i_s = \bar{x}^i_s$ initially, then that should still be true after the coupled \add{inner} MCMC updates of $x^i_s$ and $\bar{x}^i_s$.
\add{A simple way to attain this minimal requirement} is by using the same seed for the pseudorandom number generators in both \add{chains}.
Better coupling of \add{the inner} MCMC updates \add{could} further improve the \add{overall} coupling of our method.

\begin{algorithm}
\caption{Coupled particle MCMC. \label{alg:cmcmc}}
Input: Minimum number of \add{outer} MCMC steps $l$ and probability $\rho$ of using PIMH.
\begin{enumerate}
	\item
	Initialize $x_{0:S}(0)$ by running the SMC algorithm as per Step~\ref{step:pimh_proposal} of Algorithm~\ref{alg:cpimh}.
	\item
	\begin{enumerate}
		\item
		\label{step:init_pimh}
		With probability $\rho$,
		\add{generate} ${x_{0:S}(1)\cond x_{0:S}(0)}$
    	according to the PIMH algorithm,
    	for instance by running Algorithm~\ref{alg:cpimh} with
    	$x_{0:S} = \bar{x}_{0:S} = x_{0:S}(0)$,
		and
		set $\bar{x}_{0:S}(1)$ equal to the proposal $\tilde{x}_{0:S}$ from this PIMH.
		\item
		With probability $1-\rho$,
		set $\bar{x}_{0:S}(1) = x_{0:S}(0)$
		and
		\add{generate} ${x_{0:S}(1)\cond x_{0:S}(0)}$
    	according to the conditional SMC algorithm,
    	for instance by running Algorithm~\ref{alg:ccsmc} with
    	$x_{0:S} = \bar{x}_{0:S} = x_{0:S}(0)$.
		\item \add{Set $t=2$.}
	\end{enumerate}
	\item
	\add{
	While $x_{0:S}(t-1)\neq \bar{x}_{0:S}(t-1)$ or $t \leq l$:
	\begin{itemize}
	\item[(a)] Draw
	$x_{0:S}(t)\cond x_{0:S}(t-1)$
	and
	$\bar{x}_{0:S}(t)\cond \bar{x}_{0:S}(t-1)$
	using the coupled PIMH step in Algorithm~\ref{alg:cpimh} with probability $\rho$,
	and the coupled conditional SMC step in Algorithm~\ref{alg:ccsmc} with probability $1-\rho$.
	\item[(b)] Set $t=t+1$.
	\end{itemize}}
\end{enumerate}
\add{Output: Chains $\{x_{0:S}(t)\}_{t=0}^{T}$
and $\{\bar{x}_{0:S}(t)\}_{t=1}^{T}$, that meet at a random time $\tau\ge 1$, with $T = \max \{l,\tau\}$.}
\end{algorithm}

\subsection{Unbiased Monte Carlo approximation}
\label{sec:unbiased}

Algorithm~\ref{alg:cmcmc} specifies the coupled MCMC algorithm resulting from a mixture of Algorithms~\ref{alg:cpimh} and \ref{alg:ccsmc}
with \add{parameter} $\rho$.
It provides chains
$\{x(t)\}_{t=1}^T$
and $\{\bar{x}(t)\}_{t=1}^T$
which can be used to estimate expectations w.r.t.~the posterior $\pi(x)$ via ergodic averages.
By construction, $x(t-1)$ and $\bar{x}(t)$ have the same distribution for any $t\geq 1$.
Also, they meet at some time $\tau$ which is almost surely finite under the conditions given in
in Section~\ref{sec:theory} for $\rho=0,1$.
This enables unbiased Monte Carlo estimation of $\pi(h)$ as described in \eqref{eq:unbiased_simple}.

\citet[Section~2.2]{Middleton2019} propose the initialization in Step~\ref{step:init_pimh} \add{of Algorithm~\ref{alg:cmcmc}. This choice 
allows for} $\tau = 1$ as $x_{0:S}(1) = \bar{x}_{0:S}(1)$
if the PIMH step accepts.
Moreover, ${\mathrm{Pr}(\tau=1)}\geq 1/2$ if only PIMH is used, i.e.\ $\rho = 1$ \citep[Proposition~8]{Middleton2019}.
Note that a high ${\mathrm{Pr}(\tau=1)}$ does not necessarily result in low variance Monte Carlo estimation.
For instance, consider $\bar{x}(1) = x(0)$
and let
${x(1)\cond x(0)}$ follow a Metropolis-Hastings update.
Then, \add{we have that} ${\mathrm{Pr}(\tau=1)} = {\mathrm{Pr}\{\bar{x}(1) = x(1)\}}$, 
\add{i.e.~${\mathrm{Pr}(\tau=1)}$} corresponds to the Metropolis-Hastings rejection probability
\add{which will typically not be too low. However, the performance of the unbiased methodology based on coupling of such an MCMC kernel can be very poor.}

The unbiasedness of the estimator $\hat{h}_k$ in \eqref{eq:unbiased_simple} enables embarrassingly parallel computation
for independent estimates.
Consider $R$ independent runs of Algorithm~\ref{alg:cmcmc} resulting in $R$ independent
copies of the estimator denoted by
$\hat{h}_k^r$, $r=1,\dots,R$.
Then, $R^{-1} \sum_r^R \hat{h}_k^r$ is an unbiased estimator of $\pi(h)$.
Its variance decreases linearly in $R$ and can be estimated by its empirical variance.
Moreover, the unbiasedness and independence of the $R$ estimators enables the construction of
confidence intervals for $\pi(h)$
which are exact as $R\to\infty$
per the central limit theorem.

\add{In the context of particle filtering,}
\citet[Section~4]{Jacob2020} provide improvements to the unbiased estimator $\hat{h}_k$ in \eqref{eq:unbiased_simple} that reduce its variance
for each run of Algorithm~\ref{alg:cmcmc}.
\add{We apply these improvements to our setting for general posterior computation.}
Firstly, one can average over different $k$ since $\hat{h}_k$ is unbiased for any $k\geq 1$.
For any positive integers $k$ and $l$ with $k\leq l$,
this results in the unbiased estimator
\begin{equation} \label{eq:unbiased}
\begin{aligned}
	\bar{h}_k^l &= \frac{1}{l-k+1} \sum_{q=k}^l \hat{h}_q \\
	&= \frac{1}{l-k+1} \sum_{q=k}^l \left( h\{x(q)\}
	+ \sum_{t=q+1}^{\tau - 1} [h\{x(t)\} - h\{\bar{x}(t)\}] \right) \\
	&= \frac{1}{l-k+1} \sum_{t=k}^l h\{x(t)\}
	\\ &\qquad \qquad + \sum_{t=k+1}^{\tau - 1} \tfrac{\min(l-k+1,\, t-k)}{l-k+1} [h\{x(t)\} - h\{\bar{x}(t)\}],
\end{aligned}
\end{equation}
where the penultimate \add{quantity} is an ergodic average and the last \add{quantity} is a bias correction \add{term}.
The ergodic average \add{term} $(l-k+1)^{-1} \sum_{t=k}^l h\{x(t)\}$
discards the first $k-1$ steps in the chain $\{x(t)\}_{t=1}^T$ as burn-in iterations
and uses $l-k+1$ recorded iterations.

As $k$ increases, the variance of the bias correction
decreases, \add{and} the bias correction equals zero for $k\geq \tau - 1$.
Nonetheless, it is suboptimal to set $k$ very large as that increases
the variance of the ergodic average similar to when discarding too many
iterations as burn-in in MCMC.
One can pick $k$ as a high percentile of the empirical distribution of the meeting time $\tau$ from multiple runs of Algorithm~\ref{alg:cmcmc}.
Alternatively,
the empirical variance of $\bar{h}_k^l$
can be minimized via grid search over $k$
\citep[Appendix~B.2]{Middleton2019}, at the price of losing the unbiasedness of $\bar{h}_k^l$.
Parameter $l$ can be set to a large value within  computational constraints as a larger $l$
reduces the variance of the unbiased estimator in~\eqref{eq:unbiased}.

A second \add{straightforward approach for variance reduction involves} Rao-Blackwellization of  the estimator $\hat{h}_k$ over the weighted particle approximation from SMC.
So far, we have considered only the single particle selected by
Algorithm~\ref{alg:cpimh} or
Step~\ref{step:final_resample} of Algorithm~\ref{alg:ccsmc}. 
It is more efficient to use all $N$ particles via the weighted approximations defined by the pairs
$(w^{1:N},\, x_S^{1:N})$, 
$(\bar{w}^{1:N},\, \bar{x}_S^{1:N})$.
\add{That is,}
$h\{x(t)\}$, $h\{\bar{x}(t)\}$ in \eqref{eq:unbiased_simple}, \eqref{eq:unbiased} can be replaced by
$\sum_{i=1}^N w^i(t)\, h\{x_S^i(t)\} / \sum_{i=1}^N w^i(t)$
and
$\sum_{i=1}^N \bar{w}^i(t)\, h\{\bar{x}_S^i(t)\} / \sum_{i=1}^N \bar{w}^i(t)$, respectively,
where the weights and particles are given in
Step~\ref{step:resample2} of Algorithm~\ref{alg:smc} or
Step~\ref{step:final_resample} of Algorithm~\ref{alg:ccsmc}.

\section{Theoretical properties}
\label{sec:theory}

Existing analysis of coupled conditional SMC applies to our context which aims to approximate the posterior $\pi(x)$. Building on  \citet{Lee2020}, we  derive results
for Algorithm~\ref{alg:cmcmc} with $\rho=0$, such that only conditional SMC is used.
Appendix~\ref{ap:proof} contains the proofs which mostly consist of a mapping from the smoothing problem considered in \citet{Lee2020} to our context of posterior approximation.
That mapping results in the following assumptions.

\begin{assumption}
\label{assum:lik}
The likelihood is bounded. That is, \\  
$\sup_{x\in\mathcal{X}} {p(y\cond x)} < \infty$ for the observed data $y$.
\end{assumption}

\begin{assumption}
\label{assum:h}
The statistic $h:\mathcal{X}\to\mathbb{R}$, as in \eqref{eq:unbiased_simple}, is bounded. That is, $\sup_{x\in\mathcal{X}} |h(x)| < \infty$.
\end{assumption}

Many models satisfy Assumption~\ref{assum:lik}, including
the Gaussian graphical model in Section~\ref{sec:ggm}.
The likelihood from linear regression in Section~\ref{sec:horseshoe}
violates it if the number of predictors is greater than the number of observations.
The assumption relates to the boundedness of potential functions in the SMC literature, which is a common assumption
\citep[Section~2.4.1]{DelMoral2004}.
\citet[Theorem~1]{Andrieu2018} show that Assumption~\ref{assum:lik} is essentially equivalent to uniform ergodicity of the Markov chains
produced by conditional SMC in Algorithm~\ref{alg:cmcmc} with $\rho=0$.

\add{In the context of particle filtering,} \citet[Section~3]{Jacob2020} establish unbiasedness and finite variance of the estimator $\hat{h}_k$ in \eqref{eq:unbiased_simple}, like we do, but
without the boundedness in Assumption~\ref{assum:h}
and instead use an assumption jointly on $h(x)$ and the Markov chain generated by
conditional SMC.
The simpler and more restrictive Assumption~\ref{assum:h}
provides a bound on the variance of $\hat{h}_k$ in terms of the number of particles $N$
in Proposition~\ref{thm:unbiased}.

Assumptions~\ref{assum:lik} and \ref{assum:h} imply Assumptions~6 and 4 of \citet{Middleton2019}, respectively.
Therefore, the estimator $\hat{h}_k$ from Algorithm~8 with $\rho = 1$,
such that it uses only PIMH,
is unbiased and has finite variance, and $\tau <\infty$ almost surely per Proposition~8 of \citet{Middleton2019} and its proof under Assumptions~\ref{assum:lik} and \ref{assum:h}.

\begin{proposition}
	\label{thm:meeting_time}
	Suppose Assumption~\ref{assum:lik} holds.
	Consider Algorithm~\ref{alg:cmcmc} with $\rho = 0$.
	Then,
	for any number of temperatures $S\geq 0$,
	there exists a $c<\infty$ such that for any number of particles $N\geq 2$
	and any initial
	$(x_{0:S}, \bar{x}_{0:S})$,
	we have:
	\begin{enumerate}[label=(\roman*), font=\upshape]
	\item
	$
		\mathrm{Pr}(x_{0:S}' = \bar{x}_{0:S}') \geq N/(N+c)
	$
	where
	$(x_{0:S}', \bar{x}_{0:S}')$
	are distributed per coupled conditional SMC in Algorithm~\ref{alg:ccsmc}.
	\item
	$\tau < \infty$ almost surely.
	\item
	The average meeting time
	$E(\tau)\leq 2 + c/N$.
	\end{enumerate}
\end{proposition}

Proposition~\ref{thm:meeting_time} contains no conditions on the quality of the
coupling of the \add{inner} MCMC steps in Step~\ref{step:cmcmc} of Algorithm~\ref{alg:ccsmc}.
This confirms that the SMC machinery by itself enables coupling
and thus unbiased estimation for the posterior $\pi(x)$, \add{though coupling of the inner MCMC steps can lead to substantially smaller meeting times as explored empirically in Appendix~\ref{ap:add_simul}, \ad{an aspect of the method that the above theoretical results fail to capture}.}
Moreover, the meeting time $\tau$ can be made equal to 2 with arbitrarily high probability
by increasing the number of particles $N$.

Whether a sufficient increase of $N$ is practicable depends
on the SMC sampler in Algorithm~\ref{alg:smc}
which we adapt to $\pi(x)$.
For instance,
under additional assumptions,
Theorem~9 from
\citet{Lee2020}
states that
the meeting probability $\mathrm{Pr}(x_{0:S}' = \bar{x}_{0:S}')$
does not vanish if the number of particles scales as $N = \mathcal{O}(2^S S)$ where $S$ is the number of temperatures.

\begin{proposition}
\label{thm:unbiased}
Suppose Assumptions~\ref{assum:lik} and \ref{assum:h} hold.
Consider the estimator $\hat{h}_k$ of $\pi(h)$ in \eqref{eq:unbiased_simple} 
where the chains $x(t)$ and $\bar{x}(t)$ are generated by Algorithm~\ref{alg:cmcmc} with $\rho = 0$.
Denote the expectation and variance w.r.t.~the resulting distribution on $\{x(t), \bar{x}(t)\}_{t=1}^T$
by $\hat{E}(\cdot)$ and $\hat{\var}(\cdot)$, respectively.
Then,
for any number of temperatures $S\geq 0$, we have:
\begin{enumerate}[label=(\roman*), font=\upshape]
\item
For any $k\geq 1$,
$\hat{E}(\hat{h}_k) = \pi(h)$ and $\hat{\var}(\hat{h}_k) < \infty$.
\item
There exists a $c<\infty$ such that for any $k\geq 1$ and for any number of particles $N\geq 2$ in Algorithm~\ref{alg:ccsmc}, 
\begin{multline*}
	|\hat{\var}(\hat{h}_k) - \var\{h(x)\}| \\
	\leq 16 \left(\frac{N+c}{N}\right)^2
	\left(\frac{c}{N+c}\right)^{k/2}
	\sup_{x\in\mathcal{X}} |h(x) - \pi(h)|.
\end{multline*}
\end{enumerate}
\end{proposition}

Proposition~\ref{thm:unbiased} confirms the unbiasedness of $\hat{h}_k$.
Moreover, the variance penalty resulting from the bias correction sum in \eqref{eq:unbiased_simple}
vanishes as $N\to\infty$ or $k\to\infty$.
The latter confirms the role of $k$ as the number of burn-in iterations in the MCMC estimation.
The constant $c$ in Proposition~\ref{thm:unbiased} is the same as in Proposition~\ref{thm:meeting_time}.
\add{The rates of convergence implied by the upper bounds in Propositions~\ref{thm:meeting_time} and \ref{thm:unbiased} are probably conservative as noted by \citet{Lee2020}.}

\section{Simulation studies}
\label{sec:simul}

In all our applications, Algorithm~\ref{alg:FK} \add{sets up a} Feynman-Kac model \add{via a preliminary run that uses} $N_0=10^4$ particles, an ESS threshold of $\gamma_0=0.8$ to determine $\alpha_{0:S}$
and a correlation threshold of $\zeta_0=0.95$ for $m_{0:S}$. Subsequent \add{executions} of SMC algorithms 
\add{use} adaptive resampling with \add{ESS threshold} $\gamma = 0.5$,
\add{i.e.~we resample, (Step~\ref{step:resample} of Algorithm~\ref{alg:smc})
only if ESS drops below $\gamma N$.}
All empirical results use Rao-Blackwellization when computing $\bar{h}_k^l$.

\add{To evaluate algorithmic performance for a test function~$h$}, we consider the product `$\hat{\var}(\bar{h}_k^l)\times \text{time}$' of the empirical variance of $\bar{h}_k^l$
and the average computation time to obtain one $\bar{h}_k^l$
across the \add{number, $R$, of independent} runs.
This product captures the trade-off between quantity and quality of unbiased estimates:
the variance of the average of independently and identically distributed estimates $\bar{h}_k^l$
is proportional to $\hat{\var}(\bar{h}_k^l)$
and inversely proportional to the number of estimates.
The number of estimates is in turn inversely proportional to the average computation time
when working with a fixed computational budget.
\ad{Thus, `$\hat{\var}(\bar{h}_k^l)\times \text{time}$' is proportional to the variance of the final estimator for fixed computation time.}
We compute `$\hat{\var}(\bar{h}_k^l)\times \text{time}$' for $l=1,\dots,\add{l_{\max}}$,
\add{for some $l_{\max}$,} 
where $k\leq l$ is set for each $l$ such that $\hat{\var}(\bar{h}_k^l)$ is minimized.

\ad{The mixing of the outer MCMC and the meeting times both influence $\var(\bar{h}_k^l)$
with larger meeting times and worse mixing corresponding to higher variance, though the manner in which such effects combine is non-trivial.
Thus, we present mixing and meeting times in our empirical results for further insight.}
\add{Mixing is monitored via calculation of}
the integrated autocorrelation time for each run based on the chain
$\sum_{i=1}^N w^i(t)\, h\{x_S^i(t)\} / \sum_{i=1}^N w^i(t)$, $t = l/2,\dots,l$,
computed \add{via} the \texttt{R} package
\texttt{LaplacesDemon}
\citep{Statisticat2020}.

\subsection{Mixture of Gaussians}
\label{sec:mix_gauss}

The first set of results \add{are motivated by} the  model and SMC considered in
Section~B.2 of \citet{Middleton2019}.
The likelihood is
\[
	p(y\cond x) = \prod_{i=1}^{d_y} \frac{1}{d_x} \sum_{j=1}^{d_x} \mathcal{N}(y_i\cond x_\add{j}, 1)
\]
and the prior is uniform over the hypercube $[-10, 10]^{d_x}$.
We consider $d_x = 2$
and $d_y=100$, and
simulate data $y$ from $p(y\cond x)$ \add{under} true values
$x^* = (-3, 0)^\top$.
\add{Thus}, the posterior $\pi(x)\propto p(x)\, p(y\cond x)$ is multimodal.
The inner MCMC step is
a random walk \add{Metropolis, with proposed transition 
following} a Gaussian distribution with identity covariance, \add{as in \citet{Middleton2019}.}
We couple \add{the inner} MCMC across two chains by using
the same seed for the pseudorandom number generator in each of the two MCMC updates,
\add{resulting in} a common random number coupling.

The \add{scalar} statistics used to determine $m_{1:S}$ are the log-likelihood 
$f_1(x) = \log\{p(y\cond x)\}$ and $L_2$-norm $f_2(x) = \|x\|$.
We aim to estimate $\pi(h)$ for $h(x) = x_1 + x_2 + x_1^2 + x_2^2$.
We run Algorithm~\ref{alg:cmcmc} \add{for $R=1024$} times
for each $\rho=0, 1/2, 1$,
with a \add{maximum value} $l_\add{\max}=10^4$
and $N=25$ particles.

\begin{figure}
\centering
\includegraphics[width=\columnwidth]{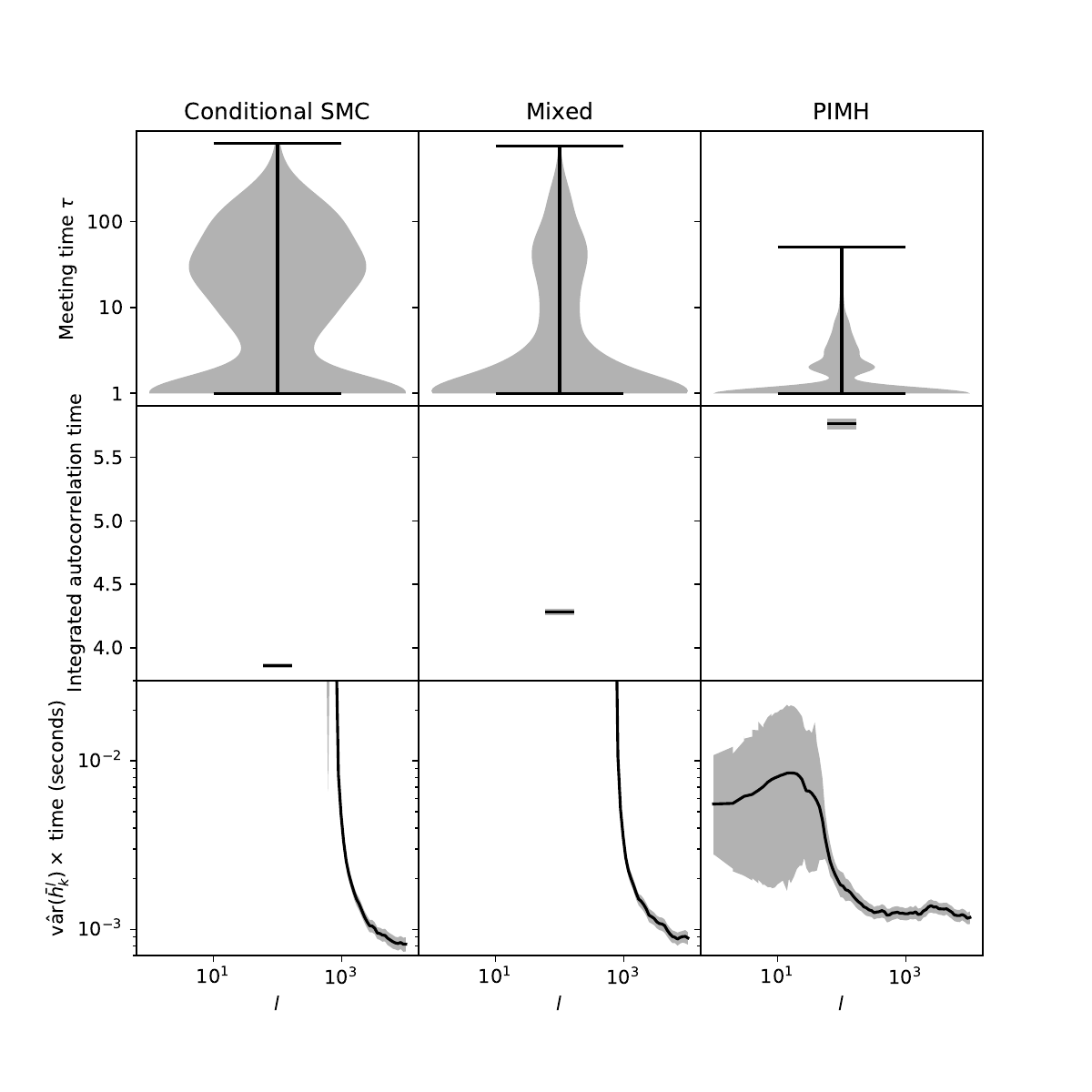}
\caption{
Results from \add{execution of} Algorithm~\ref{alg:cmcmc}
with $\rho=0$ (conditional SMC),
$\rho = 1/2$ (mixed) and $\rho = 1$ (PIMH),
for the \add{case of the} mixture of Gaussians.
The top row contains violin plots of $\log(\tau)$.
The bottom two rows show
the integrated autocorrelation time
and $\hat{\var}(\bar{h}_k^l)\times \text{time}$ as a function of $l$,
with their means and 95\% bootstrapped confidence intervals
in black and gray, respectively.
}
\label{fig:mix_gauss}
\end{figure}

Figure~\ref{fig:mix_gauss}
shows that \ad{meeting} times are lowest for Algorithm~\ref{alg:cmcmc}
with PIMH while mixing of the outer MCMC step is best with conditional SMC.
A mixed set-up with $\rho=1/2$\ad{, aimed at} 
a trade-off between these extremes\ad{, has coupling and mixing performance in between $\rho=0$ and $\rho=1$.
Nonetheless,
criterion `$\hat{\var}(\bar{h}_k^l)\times \text{time}$', defined in the second paragraph of Section~\ref{sec:simul},}
\ad{is slightly  lower for $\rho=1/2$ than for conditional SMC only for values of $l$ where PIMH is superior, a pattern observed also for scenarios described in Appendix~\ref{ap:add_simul}.
The criterion}
is lowest for conditional SMC 
owing to its better mixing but only if
Algorithm~\ref{alg:cmcmc} is run for sufficiently more iterations
than
the largest meeting times. \add{Finally, 
we
compare our methodology with the coupled HMC of \citet{Heng2019} in Appendix~\ref{ap:hmc} showing the coupled particle MCMC is more efficient in the scenario considered\ad{, though it must be noted that HMC is not well suited for multimodal targets}.}

\subsection{Horseshoe regression}
\label{sec:horseshoe}

\subsubsection{Background}

\citet{Biswas2020} recently proposed a \add{coupled} MCMC for horseshoe regression \citep{Carvalho2009}.
This section compares the coupling and its resulting unbiased estimation
with  the proposed coupled particle MCMC in Algorithm~\ref{alg:cmcmc}.

We consider the standard likelihood for linear regression
$p(y\cond x) = {\mathcal{N}(y\cond W \beta,\, \sigma^2 I_\add{d_y})}$
\add{with} $y$ a $d_y$-dimensional vector of observations,
$W$ a $d_y\times p$ design matrix and $\sigma^2$ the error variance. The main object of inference is the $p$-dimensional coefficient vector $\beta$.
The horseshoe prior on $\beta$
is one of the most popular global-local shrinkage priors \add{for} state-of-the-art Bayesian variable selection \citep{Bhadra2019}.
It is defined by
$\beta_j\cond \sigma^2,\xi,\eta_j \sim \mathcal{N}\{0,\, \sigma^2 / (\xi\eta_j) \}$
independently for $j=1,\dots,p$
where
$\xi$ is a global precision parameter
with prior $\sqrt{\xi}\sim C^+(0,1)$
and
$\eta_j$ is a local precision parameter
with prior $\sqrt{\eta_j}\sim C^+(0,1)$
independently for $j=1,\dots,p$.
Here, $C^+(0,1)$ denotes the standard half-Cauchy distribution.
That is, if $t\sim C^+(0,1)$, then the density of $t$ is
$p(t) = 2/\{\pi(1+ t^2)\}$ for $t>0$ and zero otherwise.

We follow the simulation set-up of \citet[Section~3]{Biswas2020}.
The gamma prior $1/\sigma^2 \sim \mathrm{Gamma}(1,1)$
completes the Bayesian model specification.
The elements of the matrix $W$ are sampled independently from the standard Gaussian distribution.
We draw $y\sim \mathcal{N}(W \beta_*,\, \sigma^2_* I_n)$
for some true $\beta_*$ and true error variance $\sigma^2_*$.
We set $\sigma^2_* = 8$, $d_y=100$, $p = 20$, $\beta_{*,j}=2^{(9-j)/4}$ for $j=1,\dots,10$
and $\beta_{*,j}=0$ for $j=11,\dots,p$.

\subsubsection{Coupled particle MCMC for horseshoe regression}
\label{sec:horseshoe_ccsmc}

\add{We construct our inner MCMC steps by making use of Algorithm~1 from \citet{Biswas2020}.}
Here, $x = (\beta,\eta,\sigma^2,\xi)$ since the MCMC provides a Markov chain on these parameters jointly.
Our method requires an inner MCMC step that is invariant w.r.t.~$\pi_\alpha(x) \propto p(x)\, p(y\cond x)^\alpha$, e.g.,
Step~\ref{step:mcmc} of Algorithm~\ref{alg:smc}.
Here,
$p(y\cond x)^\alpha$
equals
$p(y\cond x)$ with $n$, $y$ and $W$ replaced by
$\alpha n$, $\sqrt{\alpha} y$ and $\sqrt{\alpha} W$, respectively.
\add{Thus}, the required \add{inner} MCMC for $\pi_\alpha(x)$
follows \add{by carrying out these replacements in the inner MCMC step as developed for the full posterior} $\pi(x)$.

The \add{scalar} statistics used \add{for determining} the number of MCMC steps $m_{1:S}$ are the log-likelihood \add{function, that is}
$f_1(x) = \log\{p(y\cond x)\}$, and the number of elements in $\beta$ \add{with} absolute value
greater than $0.01$, $f_2(x) = \sum_{j=1}^p \add{\mathrm{I}[|\beta_j| > 0.01]}$.
The latter is \add{interpreted as} the number of variables
\add{chosen by} the model. 
Algorithm~\ref{alg:FK} produces $m_s = 1$ for all $s$.
This confirms the effectiveness of the \add{inner} MCMC algorithm \add{borrowed from \citet{Biswas2020}, which is accompanied by} favourable theoretical properties \add{as}
per \citet[Section~2.2]{Biswas2020}.
We compare the estimation of the posterior expectation of  $h(x) = \beta_\add{j} + \beta_\add{j}^2$, \add{$j=11$,}
obtained from different algorithms.

\subsubsection{Results}

We compare the results from coupled particle MCMC (Algorithm~\ref{alg:cmcmc}) with $\rho =1$ and $\rho=0.9$
using $N=25$ particles
to
the two-scale coupling of \citet[Section~3.2]{Biswas2020} which provides a direct coupling of \add{an} MCMC algorithm
that we use as the inner MCMC step.
Algorithm~\ref{alg:cmcmc} with $\rho=0.9$
uses this same coupling of the inner MCMC when calling Algorithm~\ref{alg:ccsmc}.
\add{We do not present the results for smaller values of $\rho$ as they are not competitive due to too large meeting times.}
The three coupled MCMCs are run $\add{R=}128$ times each for $l_\add{\max}=10^3$ iterations.
Figure~\ref{fig:horseshoe} presents the results analogous\add{ly} to Figure~\ref{fig:mix_gauss}.

\begin{figure}
\centering
\includegraphics[width=\columnwidth]{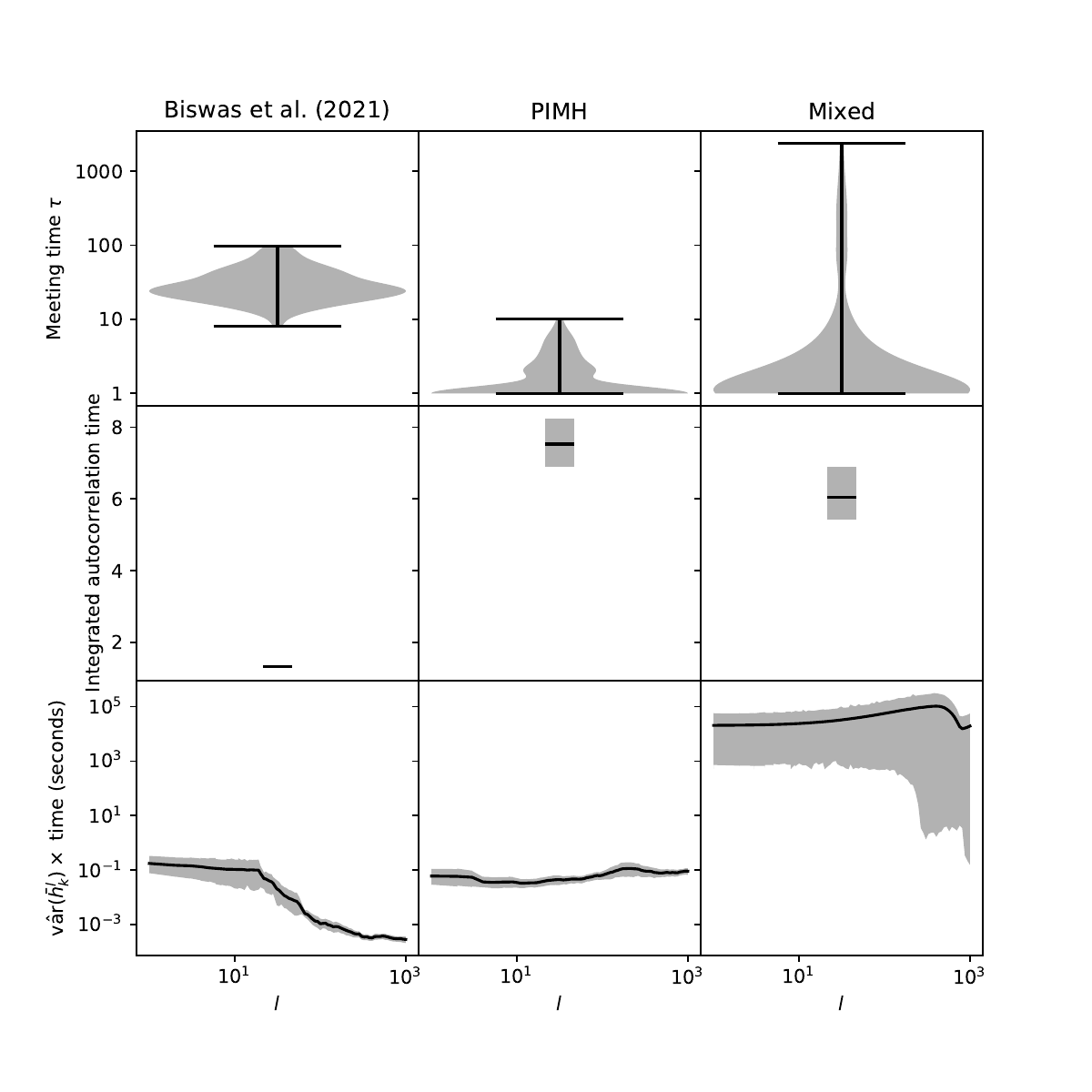}
\caption{
Results from
the coupled MCMC from \citet{Biswas2020} (left),
and
Algorithm~\ref{alg:cmcmc}
with $\rho = 1$ (PIMH)
and $\rho = 0.9$ (mixed)
for the horseshoe regression simulation.
The top row contains violin plots of $\log(\tau)$.
The bottom two rows show
the integrated autocorrelation time
and $\hat{\var}(\bar{h}_k^l)\times \text{time}$ as a function of $l$
with their means and 95\% bootstrapped confidence intervals
in black and gray, respectively.
}
\label{fig:horseshoe}
\end{figure}

Coupled particle MCMC meets in fewer iterations than the coupled MCMC from \citet{Biswas2020}.
Nonetheless,
coupled particle MCMC takes more computation time to achieve the same estimation accuracy.
One reason for this is the mixing of the coupled MCMC as shown in  Figure~\ref{fig:horseshoe}
with a much higher integrated autocorrelation time for coupled particle MCMC.
Another reason is that the computational cost of the outer particle MCMC step in Algorithm~\ref{alg:cmcmc}
is much higher than that of a single inner MCMC step.
This increased cost of particle methods \add{could} be alleviated by parallelizing across particles
if the main concern is elapsed real time rather than cumulative CPU time across CPU cores.
\ad{More generally, effective direct couplings of (inner) MCMC, like the method from \citet{Biswas2020},
would often be superior if available, as any improved coupling through SMC methods does not justify their added cost.
On the other hand,
coupled particle MCMC can be useful when effective direct coupling is not available.
}
The results for `$\hat{\var}(\bar{h}_k^l)\times \text{time}$'
show that it is not sufficient to only consider meeting times when
assessing coupled MCMC for unbiased estimation if the MCMC algorithms compared
differ in mixing or computational cost.

Figure~\ref{fig:horseshoe} suggests that conditional SMC mixes better than PIMH
but results in a worse coupling of the outer MCMC.
Here, the improvement over PIMH in mixing does not outweigh the worse coupling of conditional SMC \ad{and the associated variance of the bias correction term}
at $l=10^3$ outer MCMC iterations.
\ad{For $l$ sufficiently large,
the relative contribution of the bias correction term would vanish resulting in lower `$\hat{\var}(\bar{h}_k^l)\times \text{time}$'
for conditional SMC than for PIMH,
though such $l$ might not be computationally feasible.}

\section{Application: Gaussian graphical models}
\label{sec:ggm}

\subsection{Model}
\label{sec:ggm_model}

We consider Gaussian graphical models \citep{Dempster1972,Lauritzen1996} as an example \add{of a complex posterior on} discrete spaces. 
Coupling of MCMC kernels on the original non-extended space appears highly challenging. 
In general, generation of unbiased estimators for posterior expectations in this context is a major undertaking, and it is attempted, to the best of our knowledge, for the first time in this work. We allow the graphs to be non-decomposable, 
and present methodology that sidesteps approximation of intractable normalising constants.

The object of inference is an undirected graph $G=(V, E)$ defined by a set of nodes $V = \{1,\dots,p\}$, $p\ge 1$, and a set of edges $E\subset V\times V$.
As in \citet{Lenkoski2013}, we slightly abuse notation and write $(i,j)\in G$ for $(i,j)\in E$, i.e.~when vertices $i$ and $j$ are connected in $G$.
We aim to infer $G$ given an
$n\times p$ data matrix $Y$ with $n\ge 1$ independent rows distributed according to $\mathcal{N}(0, K^{-1})$
for a precision matrix $K\in M^+$,
where $M^+$ is the set of $p\times p$ symmetric, positive-definite matrices.
Graph $G$ constrains $K$ in that $K_{ij} = 0$
if $(i,j)\notin G$.
Thus,
$K\in M^+(G)$, where $M^+(G)\subset M^+$ is the cone of matrices $K\in M^+$ with $K_{ij} = 0$ for $(i,j)\notin G$.

\add{Under} the notation in Section~\ref{sec:cpmcmc}, \add{we have} $x=(K, G)$, $y=Y$.
Notice that we are required to include $K$ into $x$ as only then  $p(y\cond x)$ is analytically available. 
We get \add{that $p(y\cond x)$ is equal to}
${p(Y\cond K, G)} = 
(2\pi)^{-np/2}|K|^{n/2} \exp(-\tfrac{1}{2}\left<K, U\right>)$,
where $U = Y^\top Y$ is the scatter matrix
and $\left<K, U\right> = \mathrm{tr}(K^\top U)$ the trace inner product.
A conjugate prior for $K$ conditional on $G$ is the $G$-Wishart distribution $\mathcal{W}_G(\delta, D)$ \citep{Roverato2002} with density
\[
	p(K\cond  G) = \tfrac{1}{I_G(\delta, D)} |K|^{\delta/2 - 1} \exp\left( -\tfrac{1}{2} \left<K, D\right> \right),
	\quad K\in M^+(G),
\]
parametrised by  $\delta > 2$ and a positive-definite rate matrix~$D$.
The normalising constant $I_G(\delta, D)$ does not have an analytical form for general non-decomposable $G$ \citep{Uhler2018}.
We henceforth choose  $\delta=3$ in agreement with previous work \citep{Jones2005,Lenkoski2013,Tan2017}. 
Due to conjugacy, ${K\cond  G, Y} \sim \mathcal{W}_G(\delta+n,\, D^*)$
where $D^* = D+U$.
Since the objective  is inference on $G$,
one would like to compute the posterior
\begin{equation} \label{eq:target}
\begin{aligned}
	p(G\cond  Y) &\propto p(G)\, p(Y\cond  G)\\
	&= p(G) \int_{M^+(G)} p(Y\cond  K,G)\,p(K\cond  G)\, dK \\
	&\propto p(G)\, \frac{I_G(\delta+n, D^*)}{I_G(\delta, D)},
\end{aligned}
\end{equation}
with  $p(G)$ the prior on $G$.

As the normalising constant $I_G$ is not analytically available in general,
some approximation is required.
Standard approaches apply Monte Carlo or Laplace approximation of $p(G\cond Y)$ \citep{AtayKayis2005, Tan2017}.
Recent work avoids approximation of normalising constants via careful MCMC constructions. See \citet{Wang2012, Cheng2012, Lenkoski2013, Hinne2014}.
Our methodology falls in this latter direction of research.
Compared to \citet{Wang2012, Cheng2012}, our MCMC samples directly from $G$-Wishart \add{laws}, rather than preserving them, to improve mixing. Our MCMC sampler is very similar to the one in \cite{Lenkoski2013}, with the exception of sampling elements of the precision matrices directly from the full conditional distribution, rather that applying a Metropolis-within-Gibbs step.

We use a size-based prior
\citep[Section~2.4]{Armstrong2009}
for $G$.
That is, $p(G)$
is the induced marginal of a joint prior distribution on $G$
and the number of edges $n_e=|E|$ such that the prior on $G$
given $n_e$ is uniform, ${p(G\cond n_e)}\propto 1$.
For $n_e$, we \add{choose} a truncated geometric distribution with success probability $1/(p+1)$,
i.e.~$p(n_e) \propto \{p/(p+1)\}^{n_e}$,
$n_e=0,1,\dots, p(p-1)/2$. 
This choice completes the prior specification for $G$ and induces sparsity.
The success probability is such that the mean of the non-truncated geometric law
equals the number of vertices $p$.
Prior elicitation based on the sparsity constraint
that the prior mean of $n_e$ equals $p$
has also been used with independent-edge priors \citep[Section~2.4]{Jones2005}.

\subsection{Further algorithmic specifications}
\label{sec:details}

\ad{We derive an inner MCMC step in Appendix~\ref{ap:ggm_mcmc} to apply Algorithm~\ref{alg:cmcmc} to Gaussian graphical models. Here, we further detail how we use Algorithm~\ref{alg:cmcmc}.}

\subsubsection{Rejection sampling with $\alpha_0 > 0$}

We set $\alpha_0>0$
and resort to rejection sampling with the prior $p(x)$ as proposal
to sample from $\pi_{\alpha_\add{0}}(x)$
in Step~\ref{step:smc_init} of Algorithm~\ref{alg:smc}.
This reduces the computational cost for the overall SMC algorithm considerably in our applications, due to the prior $p(G)$ being diffuse:
\ad{the inner MCMC step (Algorithm~\ref{alg:mcmc})} changes at most one edge of the graph in each step
such that a large number of steps $m_1$ are required for effective diversification of particles
for small $\alpha_1$
for which $\pi_{\alpha_1}(x)\approx p(x)$.
This implementation of the $m_1$ MCMC steps then dominates the computational cost.
Setting $\alpha_0 > 0$ induces a larger $\alpha_1$ and a much smaller $m_1$.

The maximum likelihood estimate of the precision $K$
is $\arg\max_{K\ad{\in M^+}} p(Y\cond  K) = nU^{-1}$.
So,
the acceptance probability of the rejection sampler follows
as
\begin{multline*}
	\frac{p(y\cond  x^i_0)^{\alpha_0} }{ {\sup_{x\in\mathcal{X}} p(y\cond  x)^{\alpha_0}} } =
	\frac{p(Y\cond  K^i_0)^{\alpha_0}}{\sup_{K\in M^+} p(Y\cond  K)^{\alpha_0}} \\
	= (|K^i_0|\, |U|)^{n\alpha_0/2} \exp\left[\tfrac{\alpha_0}{2}\left\{ np(1 - \log n) - \left<K^i_0, U\right> \right\}\right].
\end{multline*}
We choose $\alpha_0$ to achieve an acceptance rate of $a=1/50$.
Specifically, we draw $N_0/a$ particles from the prior $p(x)$
and set $\alpha_0$ to the largest value such that $N_0$ of these particles are accepted.

\subsubsection{PIMH vs conditional SMC}

The main challenge for the class of Gaussian graphical models is the effective coupling of the outer MCMC chain so that the size of the meeting times $\tau$ is not \add{too} large.
Following the discussion in Section~\ref{sec:simul},
PIMH yields smaller $\tau$ than conditional SMC.
We therefore use Algorithm~\ref{alg:cmcmc} with $\rho=\add{1}$ here.
\add{Even so, some chains fail to couple (Section~\ref{sec:ggm_results}), suggesting that lower values of $\rho$ would result in impracticably large meeting times.}

\subsection{Data description}
\label{sec:data}

\add{We apply the Gaussian graphical models on a set of metabonomic data.}
The sample consists of $n=471$ six-year-old children from the Growing Up in Singapore Towards healthy Outcomes (GUSTO) study \citep{Soh2014}.
Measurements are taken from blood serum samples
via nuclear magnetic resonance
to obtain metabolic phenotypes.
Specifically, the levels of certain metabolites, which are low-molecular-weight molecules involved in metabolism, are measured.
Examples of metabolites are cholesterol, fatty acids and glucose.
See \citet{Soininen2009} for a description of the measurement process.
We focus on a set of 35 metabolites.
Thus, the $n\times p$ data matrix $Y$ consists of $p=35$ metabolite levels.
We quantile normalise the metabolites such that they marginally follow a standard Gaussian distribution.

\subsection{Results}
\label{sec:ggm_results}

The statistics used to determine the number of \add{inner} MCMC steps $m_s$ are the log-likelihood 
$f_1(x) = \log\{p(y\cond x)\}$ and the number of edges $f_2(x) = |E|$.
The adaptation results in $S=165$ temperatures
with $m_s$ varying from 1 to 14 and a mean of 6 for $s>1$,
\add{whereas} $m_1 = 15$.
Indeed, having $\alpha_0 > 0$ ensures that the tempered posterior $\pi_{\alpha_1}(x)$ is further away from the uninformative prior
$p(x)$ such that the number of MCMC steps for sufficient diversification of particles is limited.

Algorithm~\ref{alg:cmcmc} is run for $l=40$ MCMC steps with $N=10^3$ particles, until the chains meet or when the time budget is up,
whichever \add{happens} first.
We do this repeatedly across 16 CPU cores simultaneously and independently.
Our time budget is 72 hours.
This results in 30 finished runs of Algorithm~\ref{alg:cmcmc} and 16 incomplete runs that are cut short due to running
out of time.
\ad{The failure to meet suggests that this posterior is at the limit of what is computationally feasible with the version and numerical execution of coupled particle MCMC used in this work.}
\add{We discard the incomplete runs which introduces} a bias, though this bias decays exponentially with the time budget
as discussed in Section~3.3 of \citet{Jacob2020_2}.
Of the \add{30 finished} runs, 19 meet in $\tau = 1$ step and the largest observed $\tau$ is $10$.

\begin{figure}
\centering
\includegraphics[width=\columnwidth]{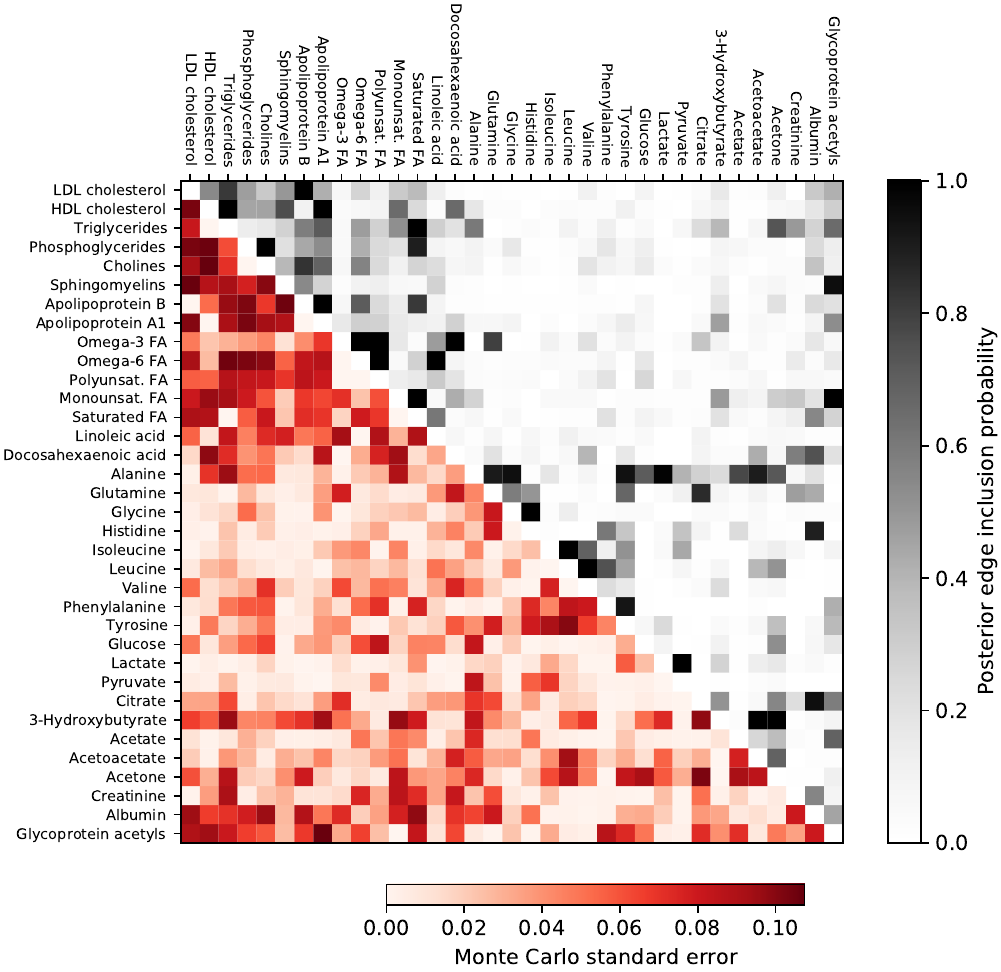}
\caption{
Posterior edge inclusion probabilities (upper triangle)
and their Monte Carlo standard errors (lower triangle)
resulting from coupled particle MCMC
for the Gaussian graphic\add{al} model fit on the metabolite data.
LDL, HDL and FA stand for low-density lipoprotein, high-density lipoprotein and fatty acids, respectively.
}
\label{fig:edge_prob}
\end{figure}

\begin{figure}
\centering
\includegraphics[width=\columnwidth]{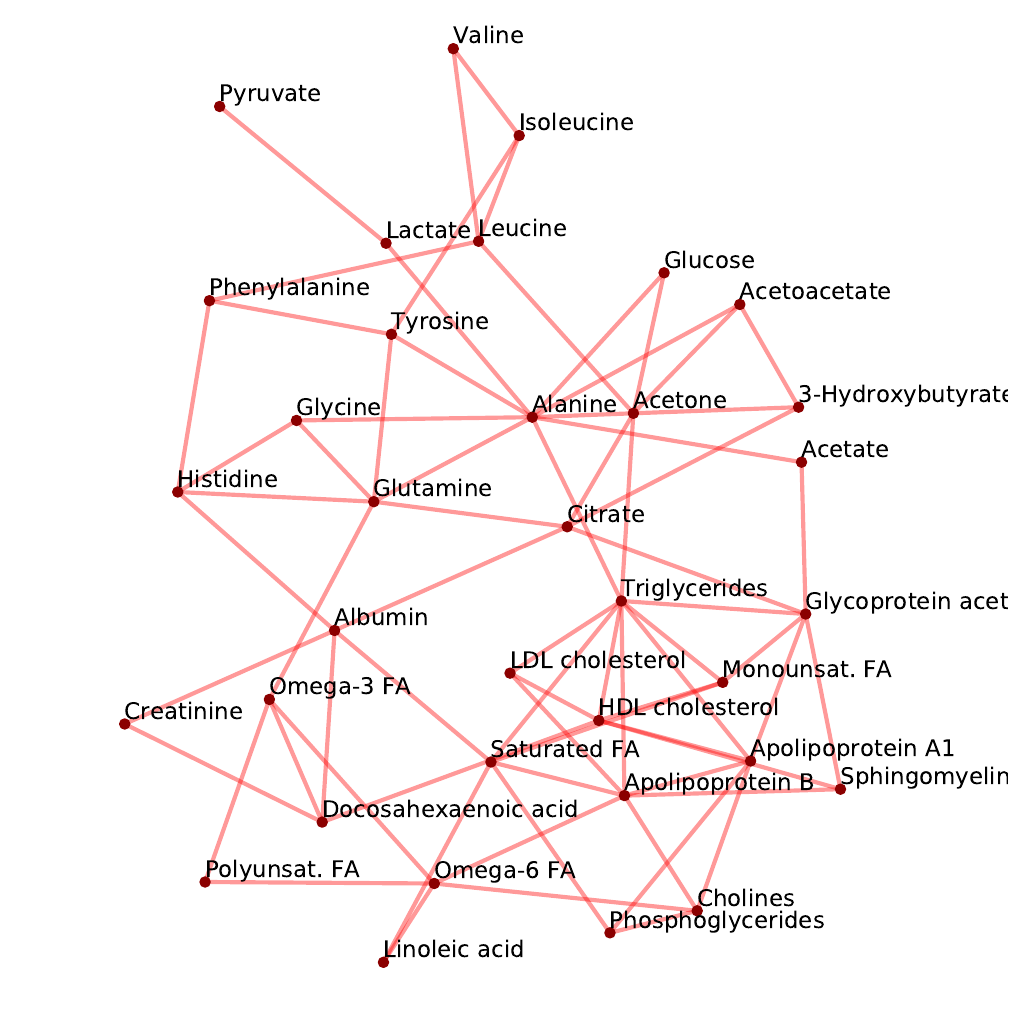}
\caption{
The median probability graph
resulting from coupled particle MCMC
for the Gaussian graphic\add{al} model fit on the metabolite data.
LDL, HDL and FA stand for low-density lipoprotein, high-density lipoprotein and fatty acids, respectively.
}
\label{fig:median_prob}
\end{figure}

We use these runs to obtain multiple unbiased estimates based on \eqref{eq:unbiased}
with $k = 7$ and $l=40$
where $h(x)$ is chosen to estimate the posterior edge inclusion probabilities, i.e.\
the posterior probability of \add{an edge being} included in the graph $G$.
Figure~\ref{fig:edge_prob} shows the estimates.
Since they result from averaging independent and unbiased estimators,
we can also estimate their Monte Carlo error.
Figure~\ref{fig:edge_prob} confirms that the chosen time budget is sufficient to achieve a usable
posterior approximation with the standard errors maxing out at 0.11.
Appendix~\ref{ap:comparison} compares these probabilities with estimates
from running SMC once with a large number of particles.
Figure~\ref{fig:median_prob} contains the median probability graph, i.e.\ the graph which includes edges with a posterior inclusion probability greater than $0.5$.

\section{Discussion}
\label{sec:discussion}

In this work we \add{propose a} coupled particle MCMC \add{methodology that} can yield effective unbiased posterior \add{approximation and enjoys} attractive theoretical properties.
Gaussian graphical models provide an example where a good coupling of the MCMC is hard while coupled
particle MCMC produces Markov chains that tend to couple quickly.
\add{ A major contribution of this manuscript is to embed recent developments in the field into a general \ad{procedure},
building  a complete and flexible computational  strategy within the area of unbiased MCMC.
Moreover, we investigate the performance of the proposed coupled particle MCMC  in challenging settings. The paper places itself within an  active and recent area of research, and our numerical applications (horseshoe regression, Gaussian graphical models) constitute  far more complex examples than those typically considered in the  literature on unbiased estimation in  state-space models. Furthermore, it is applicable, in an automatic way (in principle), to a broad class of posterior distributions, while existing literature often focuses on particular classes of targets.

The empirical results show that
using PIMH instead of conditional SMC in Algorithm~\ref{alg:cmcmc}\ad{, as controlled by the parameter $\rho$,}
provides a trade-off between coupling and MCMC mixing with
\ad{conditional SMC}
corresponding with better mixing.
Thus,
\ad{conditional SMC}
results in the most efficient unbiased estimates for number of iterations $l$ sufficiently large,
though computing them might be impracticable due to too large meeting times.
\ad{Then, PIMH provides a computationally more feasible alternative.}
\ad{Whether conditional SMC is feasible can be}
determined by preliminary and potentially incomplete runs of Algorithm~\ref{alg:cmcmc}.
\ad{Generally, conditional SMC results}
in fewer replicate unbiased estimates, which might prevent optimal use of available computing resources (e.g. parallelisation) as well as lead to less reliable estimates of Monte Carlo error.
\ad{We additionally consider a mixture of PIMH and conditional SMC. We do not find it to be empirically superior in the scenarios considered and existing theoretical results do not immediately apply}.
\ad{Such} mixture 
is similar in spirit to the use of both
a coupled Metropolis-Hastings step and a coupled HMC step
in \citet{Heng2019}
where that mixture
also provides a trade-off between coupling and MCMC mixing.}

Major empirical and theoretical improvements in coupling have been reported for
the ancestral and backward sampling modifications to \add{the standard version of} conditional SMC \citep{Chopin2015,Jacob2020,Lee2020}.
For instance,
\citet[Theorem~11]{Lee2020} show that the one-step meeting probability analysed in Proposition~\ref{thm:meeting_time}
does not vanish if $N=\mathcal{O}(S)$ for conditional SMC with backward sampling.
Unfortunately, these modifications can usually not be used in our scenario as the
transition densities in the Feynman-Kac model resulting from taking $m_s$ inner MCMC steps are often intractable. 

The use of SMC adds computational cost compared to running MCMC by itself.
For instance, a single step in the outer MCMC defined by PIMH or conditional SMC
involves $N \sum_{s=1}^S m_s$ inner MCMC steps.
This extra computational cost is amenable to parallelization due to the embarrassingly parallel nature
of the particles in SMC.
Moreover, the extra computational effort results in more accurate Monte Carlo estimates
due to Rao-Blackwellization across the $N$ particles.
One of the main contributions of this work is to gain further insight\add{s} into these novel computational strategies. For instance, \add{the numerical experiments in}
Section~\ref{sec:simul} \add{indicate} that the shape of the posterior at hand \add{determines the efficiency in the performance of  SMC-based coupling against alternative approaches.

\begin{acknowledgements}
We thank the referees for many useful suggestions that helped to greatly improve the content of the paper.

This version of the article has been accepted for publication, after peer review (when applicable) but is not the Version of Record and does not reflect post-acceptance improvements, or any corrections. The Version of Record is available online at: \url{http://dx.doi.org/10.1007/s11222-022-10093-3}. Use of this Accepted Version is subject to the publisher’s Accepted Manuscript terms of use \url{https://www.springernature.com/gp/open-research/policies/accepted-manuscript-terms}.
\end{acknowledgements}}

%
%

\bibliographystyle{spbasic}      
\bibliography{graph}   


\appendix

\section{Systematic resampling}
\label{ap:sys_sampling}

Algorithms~\ref{alg:sys_sampling} through \ref{alg:coup_cond_sys_sampling}
detail the systematic resampling methods used for the empirical results
derived from Algorithm~\ref{alg:cmcmc}.
They involve the floor function denoted by $\lfloor x\rfloor$, i.e., $\lfloor x\rfloor$ is the largest integer for which $\lfloor x\rfloor \leq x$.

\begin{algorithm}
\caption{Systematic resampling.
\label{alg:sys_sampling}}
Input: Probability vector $p_{1:N}$ and $U\in[0, 1]$.
\begin{enumerate}
	\item
	Compute the cumulative sums $v_i = N\sum_{j=1}^i p_i$, $i=1,\dots,N$.
	\item
	Let $j = 1$.
	\item
	For $i = 1,\dots,N$,
	\begin{enumerate}
    	\item
		While $v_j < U$, do $j = j+1$.
		\item
		Let $A_i = j$ and $U = U + 1$.
	\end{enumerate}
\end{enumerate}
Output:
A vector of $N$ random indices $A$ such that,
if the input $U$ \add{follows} $\mathcal{U}(0, 1)$,
then
the expectation of the frequency of any index $i$
equals $Np_i$.
\end{algorithm}

\begin{algorithm}
\caption{\citep[Algorithm~4]{Chopin2015} Conditional systematic resampling. \label{alg:cond_sys_sampling}}
Input: Probability vector $p_{1:N}$.
\begin{enumerate}
	\item
	Compute $r=Np_1 - \lfloor Np_1 \rfloor$ and sample
	\[
	U \sim
	\begin{cases}
		\mathcal{U}(0, Np_1), &r \leq 0 \\
		\frac{r(\lfloor Np_1 \rfloor + 1)}{Np_1}\, \mathcal{U}(0, r) + \frac{Np_1 - r(\lfloor Np_1 \rfloor + 1)}{Np_1}\, \mathcal{U}(r, 1),
		&r > 0
	\end{cases}.
	\]
	\item
	Obtain a vector of indices $B$ by running Algorithm~\ref{alg:sys_sampling} with $p_{1:N}$ and $U$ as input.
	\item
	Draw $A$ uniformly from all cycles of $B$ that yield $A_1 = 1$.
\end{enumerate}
Output:
A vector of $N$ random indices $A$
with $A_1 = 1$.
\end{algorithm}

\begin{algorithm}
\caption{Coupled conditional systematic resampling. \label{alg:coup_cond_sys_sampling}}
Input: Probability vectors $p_{1:N}$, $\bar{p}_{1:N}$.
\begin{enumerate}
	\item
	Generate $B$ and $\bar{B}$ by running the first two steps of Algorithm~\ref{alg:cond_sys_sampling}
	with $p_{1:N}$ and $\bar{p}_{1:N}$ as input, respectively, using the same random numbers for each.
	\item \label{step:greedy_joint}
	Construct a joint probability distribution on the cycles $A$ and $\bar{A}$ of $B$ and $\bar{B}$
	for which $A_1=\bar{A}_1=1$
	such that the marginal distributions are uniform:
	\begin{enumerate}
		\item
	    Calculate the overlap $|\{i\in 1:N\cond A_i = \bar{A}_i\}|$ for each pair of cycles with $A_1=\bar{A}_1=1$.
	    \item
		Iteratively assign the largest probability afforded by the constraint of uniform marginals
		to the pair with the highest overlap.
	\end{enumerate}
	\item
	Draw $A$ and $\bar{A}$ from this joint distribution on cycles of $B$ and $\bar{B}$.
\end{enumerate}
Output:
Vectors $A$, $\bar{A}$ of $N$ random indices with $A_1=\bar{A}_1 = 1$.
\end{algorithm}

\section{Proofs for Section~\ref{sec:theory}}
\label{ap:proof}

Our results derive from \citet{Lee2020}.
They consider a smoothing set-up which maps to our context of approximating a general posterior $\pi(x)$ using
adaptive SMC.
Specifically, their target density is \citep[Equation~1]{Lee2020}
\begin{equation} \label{eq:smoothing_target}
	\Pi(x_{0:S}) \propto M_0(x_0)\, G_0(x_0)\ \prod_{s=1}^S M_s(x_{s-1}, x_s)\, G_s(x_{s-1},x_s).
\end{equation}
In our context,
\add{the term} $M_0(x)=\pi_{\alpha_0}(x)$ is a tempered posterior, \add{the term}
$G_0(x) = {p(y\cond x)}^{\alpha_1 - \alpha_{0}}$ a tempered likelihood,
$M_s(x_{s-1}, x_s)$ the density of the Markov transition starting at $x_{s-1}$ resulting from the $m_s$ MCMC steps which are invariant w.r.t.~$\pi_{\alpha_s}(x)$ in
Step~\ref{step:mcmc} of Algorithm~\ref{alg:smc} for $s=1,\dots,S$,
$G_s(x_{s-1}, x_s) = {p(y\cond x_s)}^{\alpha_{s+1} - \alpha_{s}}$ a tempered likelihood for $s=1,\dots,S-1$,
and $G_S(x_{S-1}, x_S) = {p(y\cond x_S)}^{1 - \alpha_{S}}$ a tempered likelihood.
Then, the coupled conditional particle filter in Algorithm~2 of \citet{Lee2020} reduces to the coupled conditional SMC in our Algorithm~\ref{alg:ccsmc}.
Thus, the results in \citet{Lee2020} apply to Algorithm~\ref{alg:ccsmc}.

\subsection{Proof of Proposition~\ref{thm:meeting_time}}

	Since $G_s(x_{s-1}, x_s)$ does not depend on $x_{s-1}$, we can write $G_s(x_{s-1}, x_s) = G(x_s)$ for $s=1,\dots,S$
	as in Section~2 of \citet{Lee2020}.
	Assumption~\ref{assum:lik}, that ${p(y\cond x)}$ is bounded, implies
	that $G_s(x_s)$ is bounded for $s=0,\dots,S$,
	which is Assumption~1 in \citet{Lee2020}.
	Therefore, Theorem~8 of \citet{Lee2020} provides
	$
		\mathrm{Pr}(x_{0:S}' = \bar{x}_{0:S}') \geq N/(N+c).
	$
	
	Part~(iii) follows similarly to the proof for Theorem~10(iii) of \citet{Lee2020}:
	We have $\mathrm{Pr}(\tau > t) \leq \{1 - N/(N+c)\}^{t-1}$ for $t\geq 1$.
	Therefore,
	\[
	\begin{aligned}
		E(\tau) = \sum_{t=0}^\infty \mathrm{Pr}(\tau > t)
		&\leq 1 + \sum_{t=1}^\infty \mathrm{Pr}(\tau > t) \\
		&\leq 1 + \sum_{t=1}^\infty \left(1 - \frac{N}{N+c} \right)^{t-1}
		= 2 + \frac{c}{N},
	\end{aligned}
	\]
	where the last equality follows from the geometric series formula\\  
	$\sum_{t=0}^\infty (1 - r)^t = 1/r$ for $|r|<1$.
	Part~(iii) implies Part~(ii). \qed

\subsection{Proof of Proposition~\ref{thm:unbiased}}

Theorem~10 of \citet{Lee2020} provides results for a statistic
that we denote by
$h_{0:S}: \mathcal{X}^{S+1}\to\mathbb{R}$.
Consider $h_{0:S}$ defined by $h_{0:S}(x_{0:S})=h(x_S)$ where $h:\mathcal{X}\to\mathbb{R}$
is our statistic of interest.
Then, $h_{0:S}$ is bounded by Assumption~\ref{assum:h}.
The marginal distribution of $x_S$ under the density on $x_{0:S}$ in \eqref{eq:smoothing_target}
is our posterior of interest $\pi(x)$.
Consequently, the results for $h_{0:S}$ in Theorem~10 of \citet{Lee2020}
provide the required results for $h$. \qed

\add{
\section{Comparison with coupled HMC}
\label{ap:hmc}

The coupled HMC method of \citet{Heng2019}
provides an alternative to coupled particle MCMC
for unbiased posterior approximation
if the posterior is amenable to HMC.
The latter
typically requires
$\mathcal{X} = \mathbb{R}^{d_x}$
and that the posterior is continuously differentiable.
Here, we apply coupled HMC to the posterior considered in
Section~\ref{sec:mix_gauss}
with a slight modification to make it suitable for HMC:
the uniform prior over the hypercube $[-10, 10]^{d_x}$
is replaced by the improper prior $p(x)\propto 1$ for $x\in \mathbb{R}^{d_x}$
to ensure differentiability.
The set-up of coupled HMC follows Section~5.2 of \citet{Heng2019}
with the following differences.
The leap-frog step size is set to 0.1 instead of 1 as
the resulting MCMC failed to accept with the latter.
We do not initialize both chains independently
but instead set $\bar{x}(1)= x(0)$
as in Algorithm~\ref{alg:cmcmc}
since we found that this change reduces meeting times.
We use code from
\url{https://github.com/pierrejacob/debiasedhmc}
to implement the method from \citet{Heng2019}.

\begin{figure}
\add{
\centering
\includegraphics[width=\columnwidth]{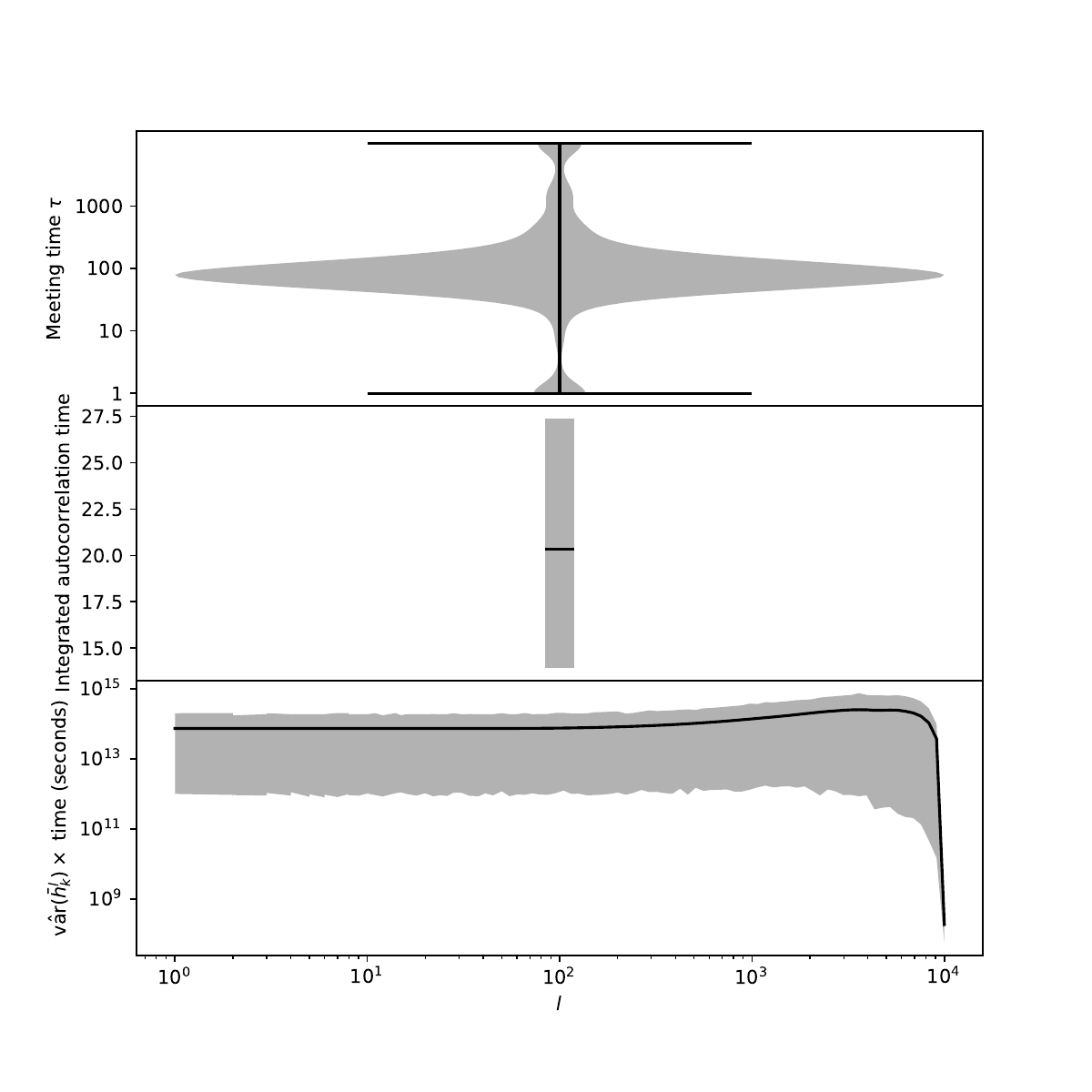}
\caption{
Results from execution of coupled HMC from \citet{Heng2019}
for the case of the mixture of Gaussians, under the adjustments described in Appendix~\ref{ap:hmc}.
The top row contains a violin plot of $\log(\tau)$.
The bottom two rows show
the integrated autocorrelation time
and $\hat{\var}(\bar{h}_k^l)\times \text{time}$ as a function of $l$,
with their means and 95\% bootstrapped confidence intervals
in black and gray, respectively.
}
\label{fig:hmc}
}
\end{figure}

Figure~\ref{fig:hmc} presents the results analogously to Figure~\ref{fig:mix_gauss}.
In terms of number of iterations, coupled HMC mixes worse and takes longer to meet than coupled particle MCMC.
These increases are not offset by a lower computational cost per iteration.
An important caveat here is that computation time depends on the implementation, and here coupled HMC is implemented using an \texttt{R} package and coupled particle MCMC in Python.

\section{Additional simulations studies}
\label{ap:add_simul}

Here, we provide some further simulation studies where the set-up is the same as in Section~\ref{sec:mix_gauss} except for the following. We consider a probability of PIMH of $\rho = 0.05$ in addition to the other values of $\rho$,
the maximum $l$ is $l_{\max}=2\cdot 10^3$
and
the number of repetitions is $R=128$.
Figure~\ref{fig:mix_gauss_N} considers different number of particles of $N$.
Figure~\ref{fig:mix_gauss_D} varies the dimensionality of the parameter $d_x$
where we use the true values
$x^* = (-3, 0, 3)^\top$
and
$x^* = (-3, 0, 3, 6)^\top$
for $d_x=3$ and $d_x=4$\ad{, respectively,}
based on the set-up in \citet[Appendix~B.2]{Middleton2019}.
Additionally,
Figure~\ref{fig:uncoupled} uses independent inner MCMC steps across both chains except for that the MCMC step is faithful to any coupling.
This contrasts with
Section~\ref{sec:mix_gauss} which uses a common random number coupling for the Metropolis-Hastings inner MCMC step.

A higher number of particles $N$ results in shorter meeting times.
Criterion `$\hat{\var}(\bar{h}_k^l)\times \text{time}$'
is lowest for larger $N$,
though beyond a certain $N$, not much improvement is gained.
\citet{Jacob2020} reach a similar conclusion when varying $N$ for coupled conditional particle filters.

Performance deteriorates with increasing dimensionality $d_x$, especially for smaller values of $\rho$.
For $d_x=4$ (Figure~\ref{fig:d4}),
the chains even often fail to meet within the maximum number of iterations of 2,000 considered for $\rho=0,0.05$.
We also see such lack of coupling in Figure~\ref{fig:uncoupled} for $\rho=0$,
suggesting that the coupling of the inner MCMC is important for good performance when working with coupled conditional SMC.
This is despite the fact that the theoretical results in Section~\ref{sec:theory} do not depend on the quality of the coupling of the inner MCMC.

For certain values of $l$,
using $\rho$ away from \ad{0} or \ad{1}
\ad{is competitive  with conditional SMC or PIMH in terms of `$\hat{\var}(\bar{h}_k^l)\times \text{time}$'  although not notably better}
than using just one of them.
\ad{The benefit of a mixture versus using only conditional SMC in terms of coupling is highlighted in Figure~\ref{fig:uncoupled} where the inner MCMC is uncoupled.}
}

\section{Inner MCMC step for Gaussian graphical models}
\label{ap:ggm_mcmc}

We set up an MCMC step with $p(x\cond  y) = p(K,G\cond Y)$ as invariant distribution.
The corresponding MCMC step for the tempered density ${p_\alpha(x\cond  y)}$, $\alpha\in(0,1]$, required for Algorithm~\ref{alg:cmcmc},
follows by replacing $n$ and $U$ by
$\alpha n$ and $\alpha U$, respectively, \add{as}
$p(y\cond x)^\alpha =\\ 
(2\pi)^{-\alpha np/2}|K|^{\alpha n/2} \exp(-\frac{1}{2}\left<K, \alpha U\right>)$.
We make use of  the algorithm for sampling from a $G$-Wishart law introduced in \citet[Section~2.4]{Lenkoski2013}.
Thus, we can sample from ${K\cond G, Y} \sim \mathcal{W}_G(\delta+n,\, D^*)$.
It remains to derive an MCMC transition that preserves $p(G\cond  Y)$,
as samples of $G$ can  be extended to $x=(K,G)$ by generating $K\cond G, Y$.

We consider the double reversible jump approach from
\citet{Lenkoski2013} 
and apply the node reordering from
\citet[Section~2.2]{Cheng2012}
to obtain an MCMC step with no tuning parameters.
The MCMC step is a Metropolis-Hastings algorithm on an enlarged space that 
bypasses the evaluation of the intractable normalisation constants
$I_G(\delta, D)$ and $I_G(\delta+n,\, D^*)$ in the target distribution~\eqref{eq:target}.
It is a combination of ideas from the PAS algorithm of  \cite{gods:01}, which avoids  the evaluation of $I_G(\delta+n,\, D^*)$, 
and the exchange algorithm of \cite{Murray2006}, which sidesteps evaluation of 
$I_G(\delta, D)$. We will give a brief presentation of the MCMC kernel that we are using
as it  does not coincide with approaches that have appeared in the literature.

\begin{figure}
\vspace{-0.1cm}
\centering
\begin{tikzpicture}[
            > = stealth, 
            shorten > = 1pt, 
            auto,
            node distance = 3cm, 
            semithick 
        ]

        \tikzstyle{every state}=[
            draw = black,
            thick,
            fill = white,
            minimum size = 8mm
        ]

        \node[state] (G) {$G$};
        \node[state] (G_tilde) [right of=G] {$\tilde{G}$};
        \node[state] (Phi) [above of=G] {$\Phi_{-f}, \Phi_{p-1,p}$};
        \node[state] (Phi_tilde) [above of=G_tilde] {$\tilde{\Phi}_{-f}, \tilde{\Phi}_{p-1,p}$};
        \node[state] (Y) [above of=Phi] {$Y$};

        \path[->] (G) edge node {} (G_tilde);
        \path[->] (G) edge node {} (Phi);
        \path[->] (G_tilde) edge node {} (Phi_tilde);
        \path[->] (Phi) edge node {} (Y);
    \end{tikzpicture}
\caption{The enlarged hierarchical model giving rise to a posterior on an extended space that will be preserved by our MCMC kernel. The construction aims at: i) removing the requirement for calculation of intractable normalising constants; ii) avoiding introducing tuning parameters. The main text defines $\Phi=(\Phi_{-f}, \Phi_{p-1,p})$ and $\tilde{\Phi}=(\tilde{\Phi}_{-f}, \tilde{\Phi}_{p-1,p})$.
}
\label{fig:model}
\end{figure}

To attain the objective of suppressing the normalising constants in the method, one works with a posterior on an extended space, defined via the directed acyclic graph in Figure~\ref{fig:model}. The left side of the graph 
gives rise to the original posterior $p(G)\, p(K\cond G)\, p(Y\cond K)$.
Denote by $\tilde{G}$ the proposed graph,  with law $q(\tilde{G}\cond G)$.
\citet{Lenkoski2013} chooses a pair of vertices $(i,j)$ in $G$, $i<j$, at random and applies a reversal, i.e.~$(i,j)\in\tilde{G}$ if and only if $(i,j)\notin G$.
The downside is that the probability of removing an edge is proportional to the number of edges in $G$, which is typically small.
Instead, we consider the method in \citet[Equation~A.1]{Dobra2011b}
that also applies the reversal,
but chooses $(i,j)$ so that the probabilities of adding and removing an edge are equal.

We reorder the nodes of $G$ and $\tilde{G}$ so that the edge that has been altered is $(p-1,p)$,
similarly to \citet[Section~2.2]{Cheng2012}.
Given $\tilde{G}$, the graph in Figure~\ref{fig:model} contains a final node that refers to the 
conditional distribution of $p(\tilde{K}\cond \tilde{G})$ which coincides with the $G$-Wishart prior  
$p(K\cond G)$. Consider the upper triangular Cholesky decomposition $\Phi$ of $K$ so that $\Phi^\top \Phi = K$.
Let $\Phi_{-f} = \Phi\setminus \Phi_{p-1,p}$. We work with the \add{map} $K \leftrightarrow \Phi=(\Phi_{-f}, \Phi_{p-1,p})$.
We apply a similar decomposition for $\tilde{K}$, and obtain the \add{map} $\tilde{K} \leftrightarrow \tilde{\Phi}=(\tilde{\Phi}_{-f}, \tilde{\Phi}_{p-1,p})$.

We can now define the target posterior on the extended space as
\begin{multline}
p\big(G, \tilde{G}, \Phi_{p-1,p}, \tilde{\Phi}_{p-1,p} \cond \Phi_{-f},  \tilde{\Phi}_{-f}, Y\big) \\  
\propto  p\big(G)\,q(\tilde{G}\cond G)\,p(\Phi\cond G)\,p(\tilde{\Phi}\cond \tilde{G})\,p(Y\cond \Phi). \label{eq:extended}
\end{multline}
Given a graph $G$,  the  current state on the extended space comprises of 
\begin{equation}
\label{eq:current}
\big(G, \tilde{G}, \Phi_{-f}, \Phi_{p-1,p},  \tilde{\Phi}_{-f}, \tilde{\Phi}_{p-1,p}\big),
\end{equation}
with $\tilde{G}\sim q(\tilde{G}\cond G)$,  and $\Phi$, $\tilde{\Phi}$ obtained from the Cholesky decomposition 
of the precision matrices $K\sim\mathcal{W}_G(\delta+n, D^\add{*})$,
$\tilde{K} \sim \mathcal{W}_{\tilde{G}}(\delta, D)$, respectively. 
Note that the rows and columns of $D$, $D^\add{*}$ have been accordingly  reordered to agree with the re-arrangement of the nodes we describe above. 
Consider the scenario with the proposed graph $\tilde{G}$ having one more edge than $G$.
Given the current state in \eqref{eq:current}, the algorithm proposes a move to the state
\begin{equation}
\label{eq:proposal}
\big(\tilde{G},G, \Phi_{-f}, \Phi^\textnormal{pr}_{p-1,p}, \tilde{\Phi}_{-f}, \tilde{\Phi}^\textnormal{pr}_{p-1,p}\big).
\end{equation}
The value $\Phi^\textnormal{pr}_{p-1,p}$ is sampled from the conditional law 
of 	${\Phi}_{p-1,p}\cond {\Phi}_{-f},  Y$.  

We provide here some justification for the above construction. The main points are the following: 
(i) the proposal corresponds to an exchange of $G\leftrightarrow \tilde{G}$, coupled with a suggested \add{value}
for the newly `freed' matrix element $\Phi^\textnormal{pr}_{p-1,p}$; (ii) from standard properties of the general exchange algorithm, switching the position of $G, \tilde{G}$ will cancel out the normalising constants of the $G$-Wishart prior from the acceptance probability; (iii) the normalising constants of the $G$-Wishart posterior never appear, as the precision matrices are not integrated out.

Appendix~\ref{ap:prop_prec} derives that
\begin{equation} \label{eq:prop_prec}
	{\Phi}_{p-1,p}\cond {\Phi}_{-f}, Y \sim \mathcal{N}\left( \frac{-D^*_{p-1,p} {\Phi}_{p-1,p-1}}{D^*_{p,p}},\, \frac{1}{D^*_{p,p}} \right)
\end{equation}
This avoids the tuning of a step-size parameter arising in the Gaussian proposal of \citet[Section~3.2]{Lenkoski2013}. 
The variable  $\tilde{\Phi}^\textnormal{pr}_{p-1,p}$ is not free, due to the edge $(p-1,p)$ assumed being removed, and is given as \citep[Equation~10]{Roverato2002}
\[
\tilde{\Phi}^\textnormal{pr}_{p-1,p} = - \sum_{i=1}^{p-2} \tilde{\Phi}_{i,p-1}\tilde{\Phi}_{ip} / \tilde{\Phi}_{p-1,p-1}
\]
The acceptance probability of the proposal is given in Step~\ref{step:accept} of the complete MCMC transition shown in Algorithm~\ref{alg:mcmc} for exponent $\epsilon=1$. In the opposite scenario when an edge is removed from $G$, 
then, after again re-ordering the nodes, the proposal $\tilde{\Phi}^\textnormal{pr}_{p-1,p}$ is sampled from 
$$\tilde{\Phi}_{p-1,p}\cond \tilde{\Phi}_{-f} \sim \mathcal{N}\left( \frac{-D_{p-1,p} \tilde{\Phi}_{p-1,p-1}}{ D_{p,p}},\, \frac{1}{D_{p,p}} \right)$$
whereas we fix
$\Phi_{p-1,p}^\textnormal{pr} = - \sum_{i=1}^{p-2} \Phi_{i,p-1}\Phi_{ip} /\Phi_{p-1,p-1}$.
The corresponding acceptance probability for the proposed move is again as in Step~\ref{step:accept}
of  Algorithm~\ref{alg:mcmc}, but now  for $\epsilon=-1$.

\begin{algorithm}
\caption{Inner MCMC \add{step} for the Gaussian graphical model. \label{alg:mcmc}}
Input: Graph $G$.
\begin{enumerate}
		\item \label{step:proposal_G}
		\begin{enumerate}
			\item
		If $G$ is complete or empty,
		sample $(i,j)$, $i<j$, uniformly from the edges in $G$ or the edges not in $G$, respectively.
		\item
		Else, sample $(i,j)$, $i<j$, as follows: w.p.~$1/2$,  uniformly from the edges in $G$;
		w.p.~$1/2$, uniformly from the edges not in $G$.
		\end{enumerate}
		\item
		\label{step:reorder}
		Reorder the nodes in $G$ so that the $i$-th and $j$-th nodes become
		the $(p-1)$-th and $p$-th nodes, respectively. Rearrange $D$ and $D^*$ accordingly.
		\item
		Let $\tilde{G}$ be as $G$ except for edge $(p-1, p)$, with $(p-1,p)\in \tilde{G}$ if and only if $(p-1,p)\notin G$.
		\item
		\label{step:rng1}
		Sample $K\sim\mathcal{W}_G(\delta+n,\, D^*)$,
		$\tilde{K} \sim \mathcal{W}_{\tilde{G}}(\delta,\, D)$.
		Compute the corresponding upper triangular Cholesky decompositions $\Phi$, $\tilde{\Phi}$.
		\item Fix $\Phi_{-f} = \Phi\setminus \Phi_{p-1,p}$ and $\tilde{\Phi}_{-f} =\tilde{\Phi}\setminus \tilde{\Phi}_{p-1,p}$.
		\label{step:rng2}
		\begin{enumerate}
			\item
			If $(p-1,p)\notin G$,
			sample the proposal 
			$$\Phi_{p-1,p}^\textnormal{pr}\cond \Phi_{-f},Y \sim \mathcal{N}( -D^*_{p-1,p} \Phi_{p-1,p-1} / D^*_{p,p},\, 1/D^*_{p,p} )$$
			and set $\tilde{\Phi}^\textnormal{pr}_{p-1,p} = - \sum_{i=1}^{p-2} \tilde{\Phi}_{i,p-1}\tilde{\Phi}_{ip} / \tilde{\Phi}_{p-1,p-1}$.\vspace{0.2cm}
			\item
			If $(p-1,p)\in G$, sample the proposal 
			$$\tilde{\Phi}^\textnormal{pr}_{p-1,p}\cond \tilde{\Phi}_{-f} \sim \mathcal{N}( -D_{p-1,p} \tilde{\Phi}_{p-1,p-1} / D_{p,p},\, 1/D_{p,p} )$$
			and set 
			$\Phi_{p-1,p}^\textnormal{pr} = - \sum_{i=1}^{p-2} \Phi_{i,p-1}\Phi_{ip} /\Phi_{p-1,p-1}$.
		\end{enumerate}
		\item \label{step:accept}
		Accept proposed move from $\big(G, \tilde{G},\Phi_{-f}, \Phi_{p-1,p}, \tilde{\Phi}_{-f}, \tilde{\Phi}_{p-1,p}\big)$ to state $\big(\tilde{G},G, \Phi_{-f}, \Phi^\textnormal{pr}_{p-1,p}, \tilde{\Phi}_{-f}, \tilde{\Phi}^\textnormal{pr}_{p-1,p}\big)$ w.p.~$1\wedge R$, for $R$ equal to
		\begin{multline*} 
			\frac{p(\tilde{G})\, q(G\cond \tilde{G})}{p(G)\, q(\tilde{G}\cond G)}
	\exp\big\{ - \tfrac{1}{2}
		\langle K^\textnormal{pr} - K, D^* \rangle
		- \tfrac{1}{2} \langle \tilde{K}^\textnormal{pr} - \tilde{K}, D \rangle
	\big\} \\
	\times  \Big[
	\tfrac{\Phi_{p-1,p-1} \sqrt{D_{p,p}}}{\tilde{\Phi}_{p-1,p-1} \sqrt{D^*_{p,p}}}
	\exp\big\{ \tfrac{D_{p,p}}{2}(\tilde{\Phi}^\textnormal{pr}_{p-1,p}-\tilde{\theta})^2 -\tfrac{D^*_{p,p}}{2} (\Phi^\textnormal{pr}_{p-1,p} - \theta)^2 \big\} \Big]^\epsilon,\label{eq:accept_prob}
		\end{multline*}
		for $\tilde{\theta} = -D_{p-1,p} \tilde{\Phi}_{p-1,p-1} / D_{p,p}$ and
		$\theta = -D^*_{p-1,p}\Phi_{p-1,p-1} / D^*_{p,p}$. Here, $K^\textnormal{pr}$ (resp.~$\tilde{K}^\textnormal{pr}$) denotes 
		the precision with upper triangular Cholesky decomposition given by the synthesis 
		of $\Phi_{-f}, \Phi^\textnormal{pr}_{p-1,p}$ (resp.~$\tilde{\Phi}_{-f}, \tilde{\Phi}^\textnormal{pr}_{p-1,p}$),
		and
		$\epsilon=-1$ if $(p-1,p)\in G$, else $\epsilon=1$.
		\item
		Revert the reordering from Step~\ref{step:reorder}. Return $\tilde{G}$ if the proposed move at Step~\ref{step:accept} is accepted, else return $G$. 
	\end{enumerate}
Output: MCMC update for graph $G$
such that the invariant distribution is the target posterior $p(G\cond  Y)$.
\end{algorithm}

\section{Proposal for precision matrices}
\label{ap:prop_prec}

This derivation is similar to Appendix~A of \citet{Cheng2012}.
Assume that the edge $(p-1,p)$ is in the proposed graph $\tilde{G}$ but not in $G$.
The prior on $\tilde{\Phi}_{p-1,p}\cond \tilde{\Phi}_{-f}$ follows from Equation~2 of \citet{Cheng2012} as
\[
	p(\tilde{\Phi}_{p-1,p}\cond \tilde{\Phi}_{-f},\tilde{G}) \propto \exp\left( -\frac{1}{2} \langle\tilde{\Phi}^\top \tilde{\Phi}, D\rangle \right).
\]
The likelihood is
\[
	p(Y\cond \tilde{K}) \propto |\tilde{K}|^{n/2} \exp\left( -\frac{1}{2} \langle\tilde{K}, U\rangle \right).
\]
Here, $|\tilde{K}|$ does not depend on $\tilde{\Phi}_{p-1,p}$ since $|\tilde{K}| = |\tilde{\Phi}|^2 = (\prod_{i=1}^p \tilde{\Phi}_{ii})^2$.
Combining the previous two displays thus yields
$p(\tilde{\Phi}_{p-1,p}\cond \tilde{\Phi}_{-f},Y) \propto
	\exp( -\langle\tilde{\Phi}^\top \tilde{\Phi}, D^*\rangle / 2)$.
Dropping terms not involving $\tilde{\Phi}_{p-1,p}$ yields \eqref{eq:prop_prec}.

\section{Comparison with SMC for the metabolite application}
\label{ap:comparison}

We compare the results in Figure~\ref{fig:edge_prob} with those from running the SMC in Algorithm~\ref{alg:smc}
with a large number of particles $N=10^5$.
Comparing Figures~\ref{fig:edge_prob} and \ref{fig:edge_prob_SMC}
shows that the results are largely the same.
The edge probabilities for which they differ substantially are harder to estimate according
to the Monte Carlo standard errors from coupled particle SMC in Figure~\ref{fig:edge_prob}.

\begin{figure}[htbp]
\centering
\includegraphics[width=\columnwidth]{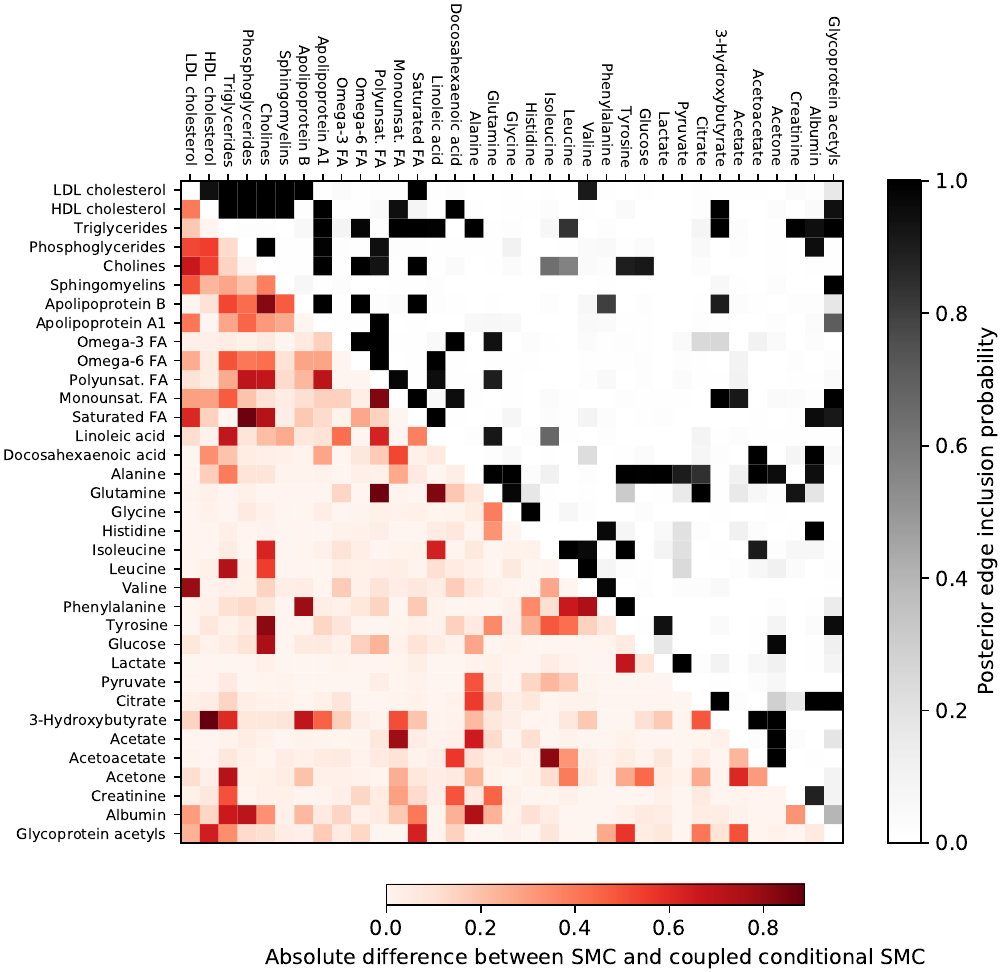}
\caption{
Posterior edge inclusion probabilities for the Gaussian graphic\add{al} model fit on the metabolite data from SMC with $N=10^5$ particles (upper triangle)
and their absolute difference from the estimates from coupled particle MCMC in Figure~\ref{fig:edge_prob} (lower triangle).
LDL, HDL and FA stand for low-density lipoprotein, high-density lipoprotein and fatty acids, respectively.
}
\label{fig:edge_prob_SMC}
\end{figure}

\begin{figure*}
\add{
   \subfloat[$d_x = 2$, $N = 25$]{
      \includegraphics[width=\columnwidth]{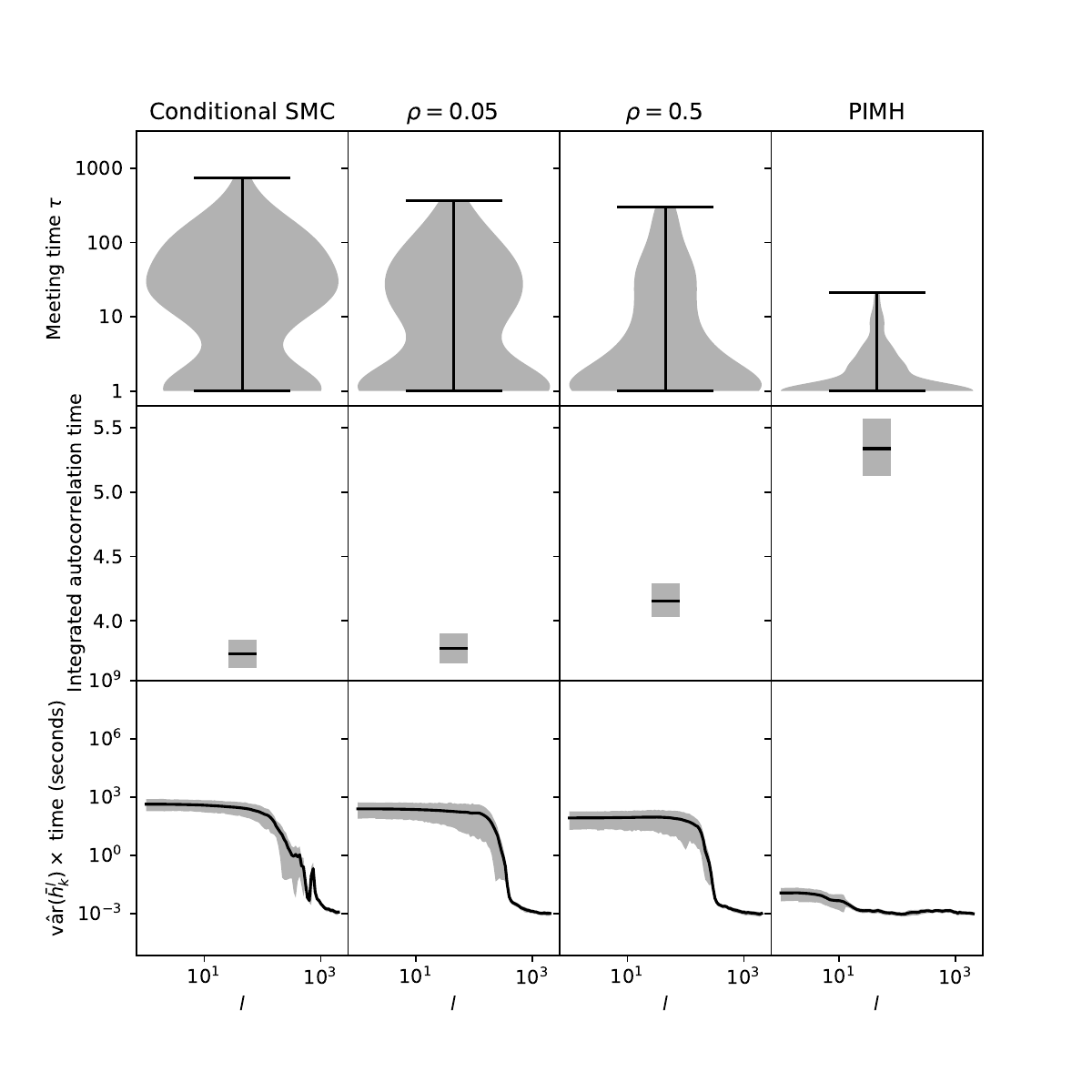}}
~
   \subfloat[$d_x = 2$, $N = 100$]{
      \includegraphics[width=\columnwidth]{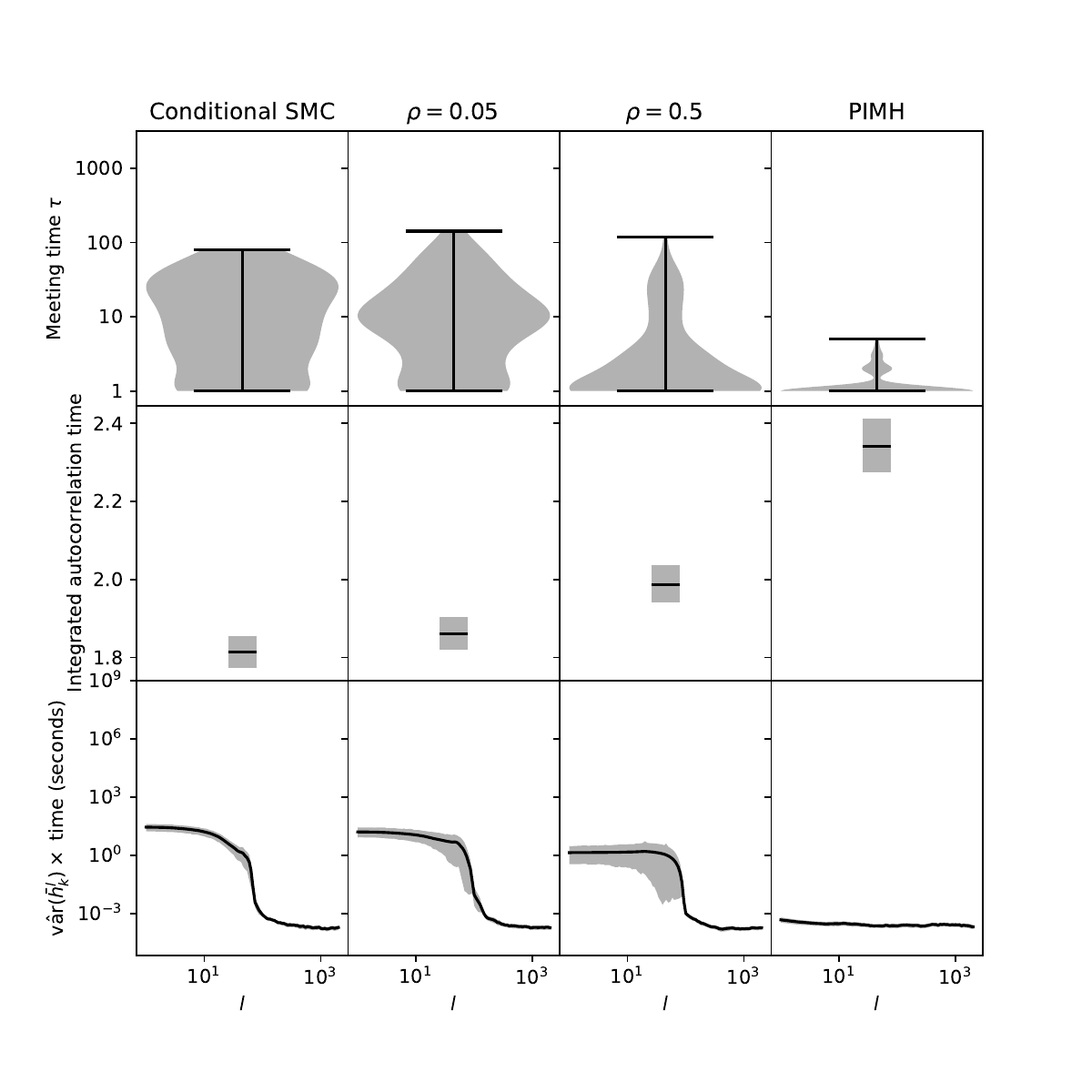}}

   \subfloat[$d_x = 2$, $N = 500$]{
      \includegraphics[width=\columnwidth]{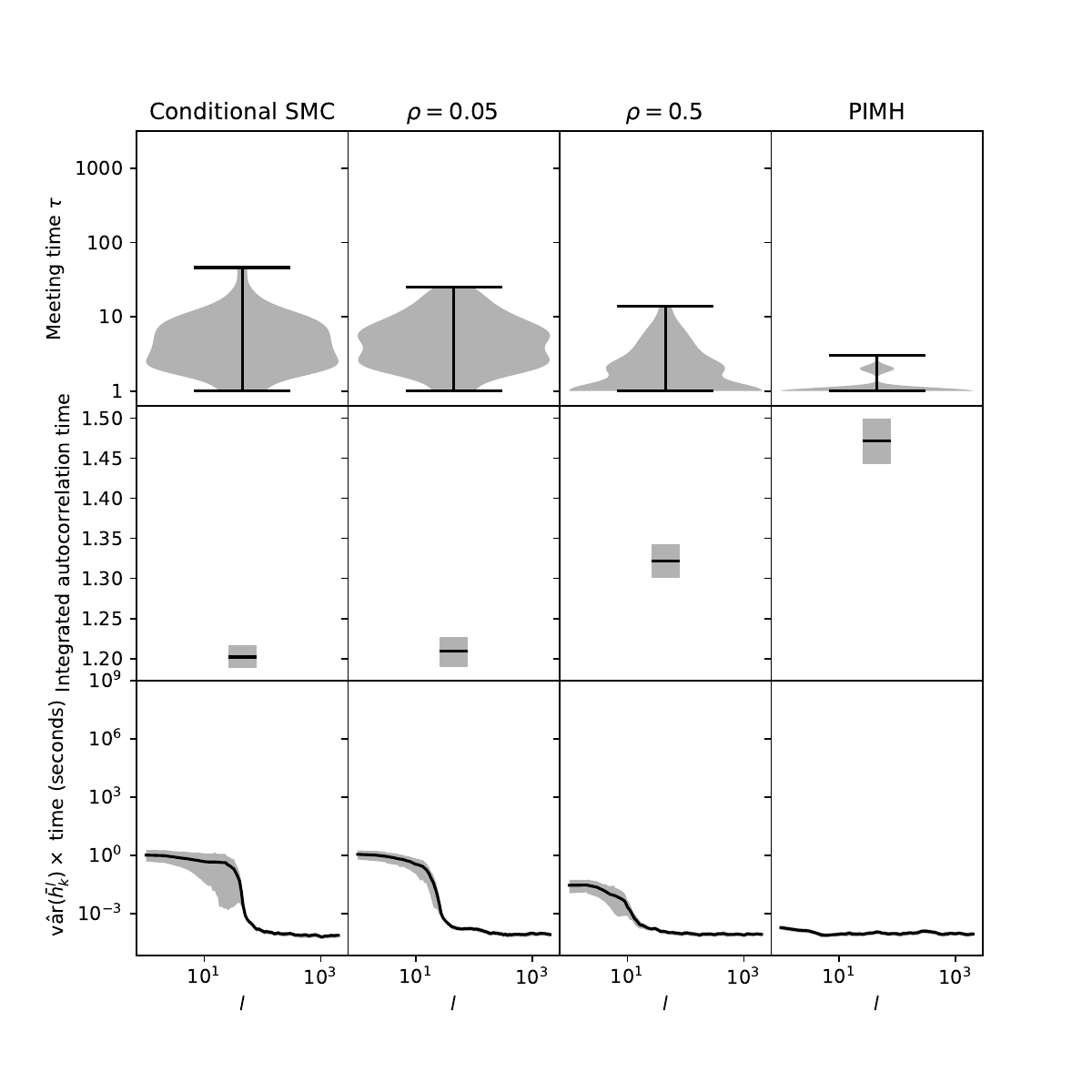}}
~
   \subfloat[$d_x = 2$, $N = 1000$]{
      \includegraphics[width=\columnwidth]{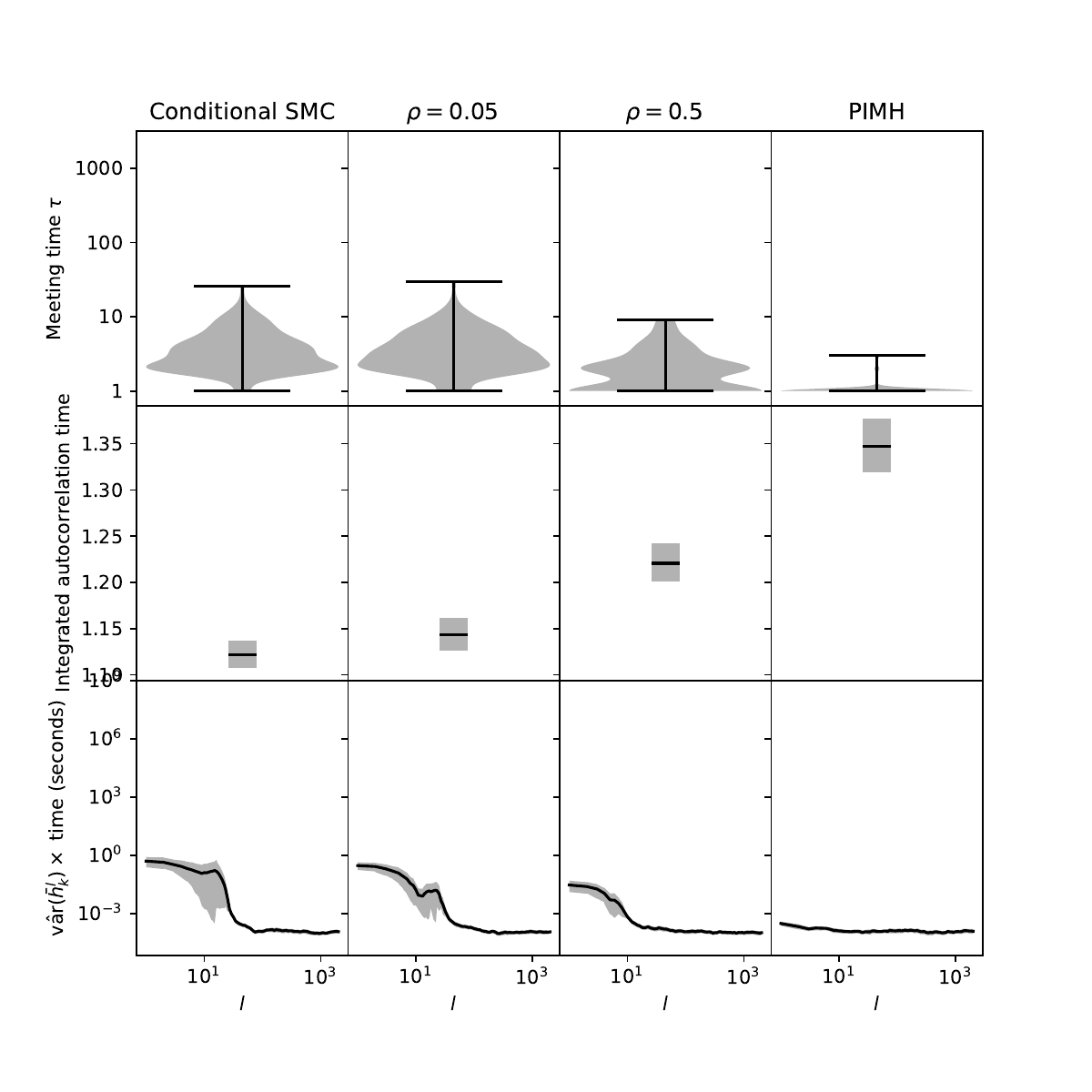}}

   \caption{
Results from execution of Algorithm~\ref{alg:cmcmc}
with $\rho=0$ (conditional SMC),
$\rho = 0.05$, $\rho = 0.5$ and $\rho = 1$ (PIMH),
for the case of the mixture of Gaussians for various number of particles $N$.
The top row contains violin plots of $\log(\tau)$.
The bottom two rows show
the integrated autocorrelation time
and $\hat{\var}(\bar{h}_k^l)\times \text{time}$ as a function of $l$,
with their means and 95\% bootstrapped confidence intervals
in black and gray, respectively.
   }
\label{fig:mix_gauss_N}
}
\end{figure*}

\begin{figure*}
\add{
   \subfloat[$d_x = 2$, $N = 500$]{
      \includegraphics[width=\columnwidth]{figures/mix_of_Guassians_D2N500.pdf}}
~
   \subfloat[$d_x = 2$, $N = 500$, uncoupled inner MCMC]{ \label{fig:uncoupled}
      \includegraphics[width=\columnwidth]{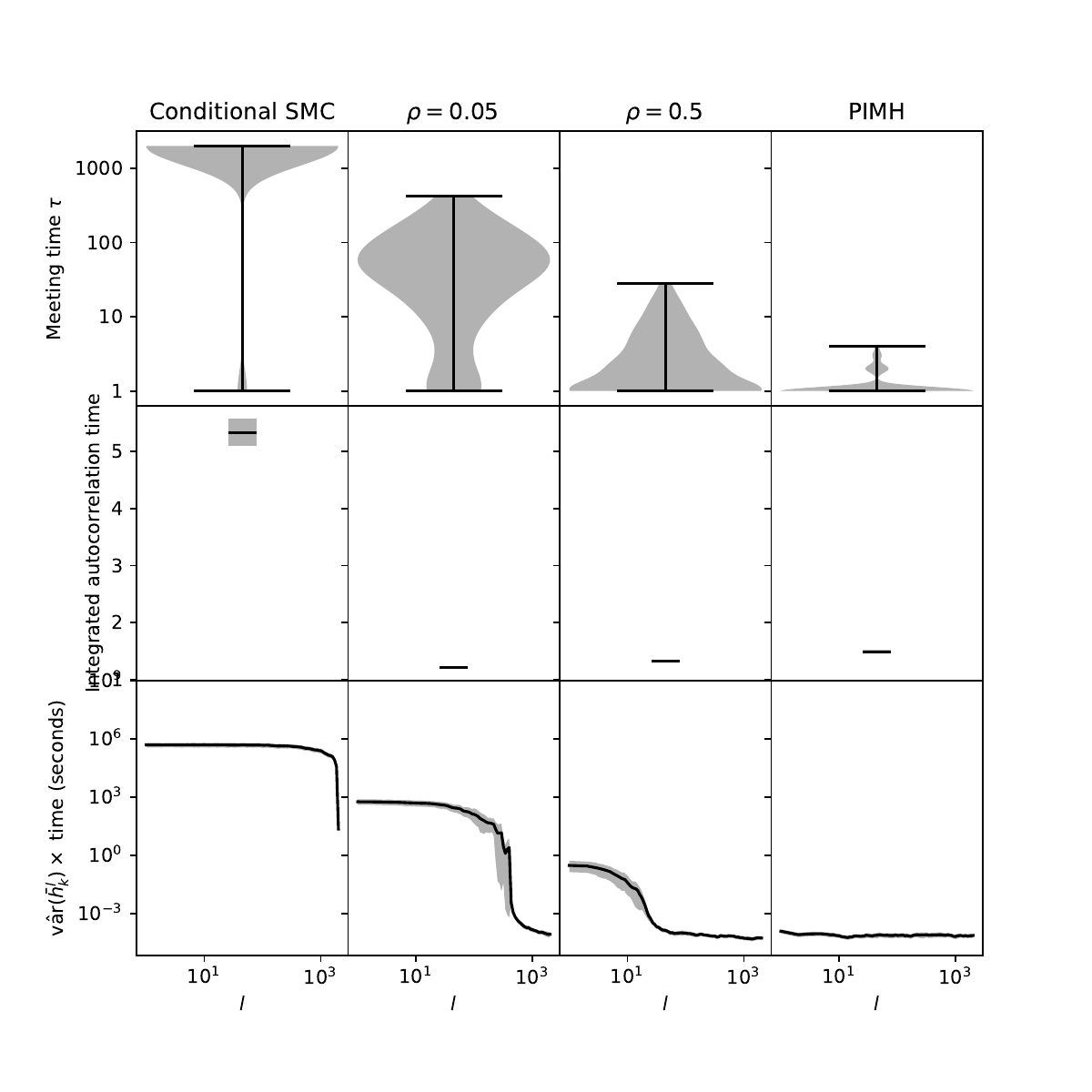}}

   \subfloat[$d_x = 3$, $N = 500$]{
      \includegraphics[width=\columnwidth]{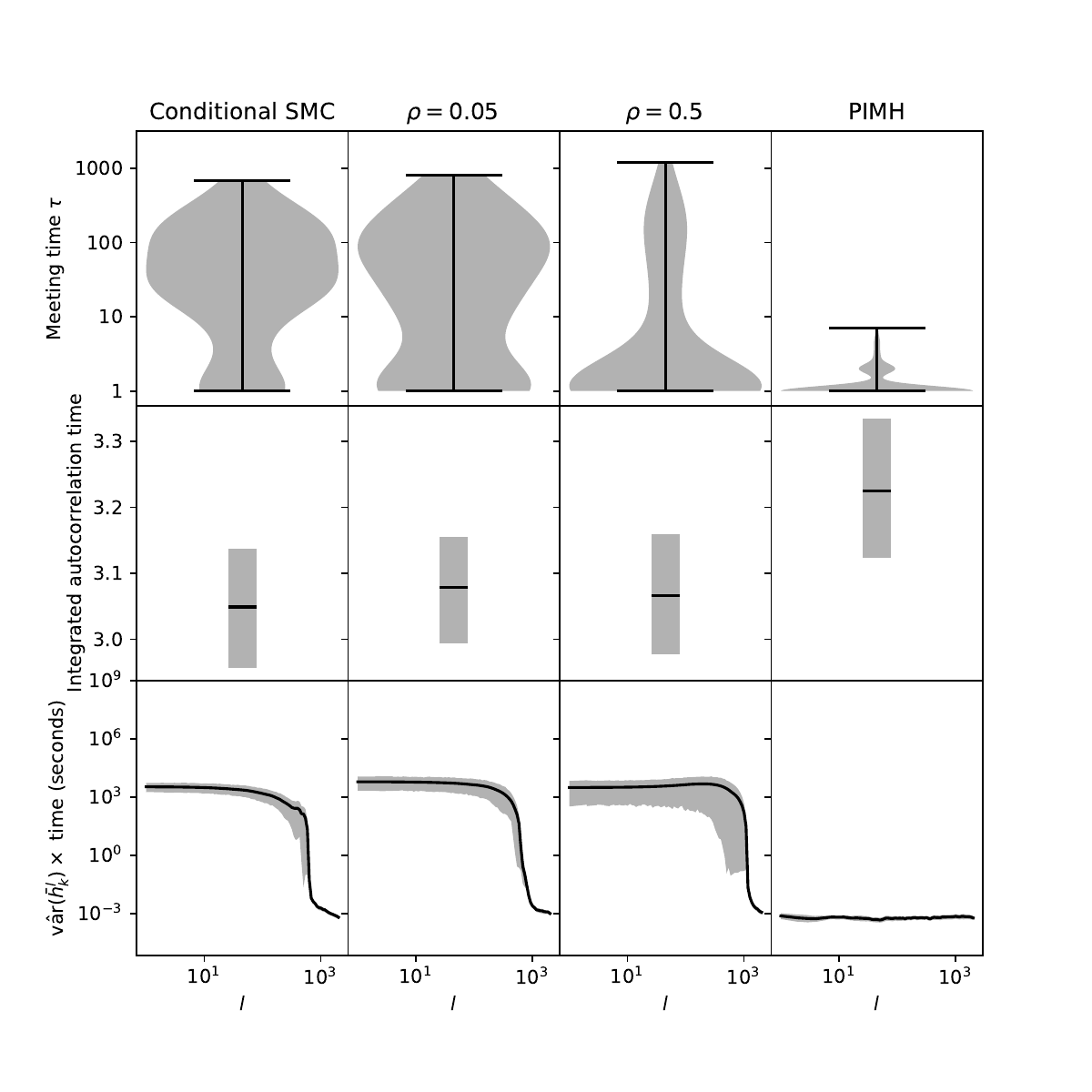}}
~
   \subfloat[$d_x = 4$, $N = 500$]{ \label{fig:d4}
      \includegraphics[width=\columnwidth]{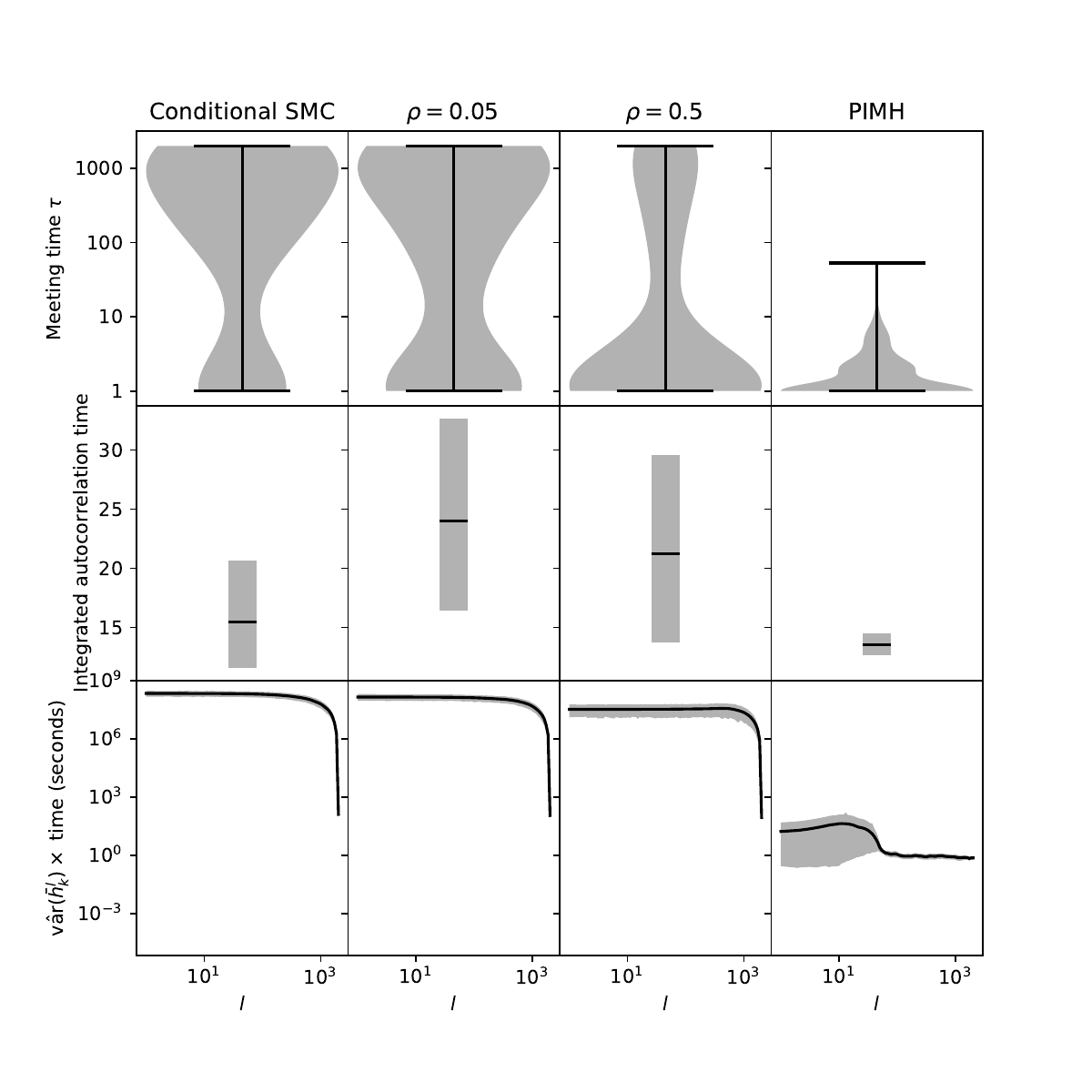}}

   \caption{
Results from execution of Algorithm~\ref{alg:cmcmc}
with $\rho=0$ (conditional SMC),
$\rho = 0.05$, $\rho = 0.5$ and $\rho = 1$ (PIMH),
for the case of the mixture of Gaussians for various parameter dimensionalities $d_x$ and, in (b), for when the inner MCMC steps are not coupled beyond being faithful.
The top row contains violin plots of $\log(\tau)$.
The bottom two rows show
the integrated autocorrelation time
and $\hat{\var}(\bar{h}_k^l)\times \text{time}$ as a function of $l$,
with their means and 95\% bootstrapped confidence intervals
in black and gray, respectively.
   }
\label{fig:mix_gauss_D}
}
\end{figure*}

\end{document}